\documentclass[conference]{IEEEtran}

\usepackage{xspace}
\usepackage{threeparttable}
\usepackage{subcaption}
\usepackage{paralist}
\usepackage{wrapfig}
\usepackage{multirow}
\usepackage{listings}
\usepackage{comment}
\usepackage{makecell}
\usepackage{amsmath}

\usepackage{amssymb}
\usepackage{array}
\usepackage{pifont}
\usepackage{enumitem}
\usepackage{algorithm}
\usepackage{algorithmic}
\usepackage{float}
\usepackage{multicol}
\usepackage{eqparbox}
\usepackage{colortbl}
\usepackage{tcolorbox}
\usepackage{tabularx}
\usepackage{tikz}
\usepackage{booktabs}
\usepackage{url}
\usepackage{amsthm}
\usepackage[normalem]{ulem}
\usepackage{balance}
\usetikzlibrary{tikzmark, arrows.meta}

\newcommand{\cmark}{\ding{51}}%
\newcommand{\xmark}{\ding{55}}%
\newcommand{\whiteding}[1]{\ding{\numexpr171+#1\relax}}

\setenumerate[1]{itemsep=0pt,partopsep=0pt,parsep=\parskip,topsep=5pt}
\setitemize[1]{itemsep=0pt,partopsep=0pt,parsep=\parskip,topsep=5pt}
\setdescription{itemsep=0pt,partopsep=0pt,parsep=\parskip,topsep=5pt}

\newcommand{\tabincell}[2]{\begin{tabular}{@{}#1@{}}#2\end{tabular}}
\newcommand{\bheading}[1]{\noindent\textbf{#1}}
\newcommand{\iheading}[1]{\noindent\textit{#1}}
\newcolumntype{?}{!{\vrule width 1pt}}

\newcommand{\code}[1]{\textbf{#1}}
\newcommand{\Seal}{\ensuremath{\code{Seal}}}
\newcommand{\UnSeal}{\ensuremath{\code{UnSeal}}}

\newcommand{\IncTC}{\ensuremath{\code{IncTC}}}
\newcommand{\ReadTC}{\ensuremath{\code{ReadTC}}}

\newcommand{\codev}[1]{{\small\textit{#1}}}
\newcommand{\votedFor}{\codev{votedFor}\xspace}
\newcommand{\currentTerm}{\codev{currentTerm}\xspace}

\newcommand{\leader}{\codev{leader}\xspace}
\newcommand{\currentEpoch}{\codev{currentEpoch}\xspace}
\newcommand{\acceptedEpoch}{\codev{acceptedEpoch}\xspace}
\newcommand{\raftMeta}{\codev{raftMeta}\xspace}
\newcommand{\msg}[1]{\textsc{#1}}

\newcounter{note}[section]

\definecolor{darkgreen}{rgb}{0.0, 0.5, 0.0}

\colorlet{Mycolor1}{green!10!orange!90!}

\newcommand{\ssecref}[1]{\mbox{\S\ref{#1}}\xspace}
\newcommand{\figref}[1]{\mbox{Fig.~\ref{#1}}}

\newcommand{\ignore}[1]{}

\newcommand{\ie}{\textit{i.e.}\xspace}
\newcommand{\eg}{\textit{e.g.}\xspace}

\newcommand{\etal}{\textit{et al.}\xspace}
\newcommand{\sysname}{\textsc{Chimera}\xspace}

\newcounter{packednmbr}

\newenvironment{packeditemize}{
\begin{list}{$\bullet$}{
\setlength{\labelwidth}{0pt}
\setlength{\itemsep}{2pt}
\setlength{\leftmargin}{\labelwidth}
\addtolength{\leftmargin}{\labelsep}
\setlength{\parindent}{0pt}
\setlength{\listparindent}{\parindent}
\setlength{\parsep}{1pt}
\setlength{\topsep}{1pt}}}{\end{list}}

\makeatletter
\@ifundefined{theorem}{\newtheorem{theorem}{Theorem}}{}
\@ifundefined{lemma}{\newtheorem{lemma}{Lemma}}{}
\@ifundefined{corollary}{}{}
\@ifundefined{proposition}{}{}
\@ifundefined{example}{}{}
\@ifundefined{remark}{}{}
\@ifundefined{definition}{\newtheorem{definition}{Definition}}{}
\makeatother

\begin{document}

\title{\sysname: Protocol-Aware Recovery for Confidential BFT Consensus}

\author{
\IEEEauthorblockN{
Tong Liu\IEEEauthorrefmark{1},
Xiaoqing Wen\IEEEauthorrefmark{2},
Ziwei Zhou\IEEEauthorrefmark{3},
Si Liu\IEEEauthorrefmark{4},
Jianyu Niu\IEEEauthorrefmark{5},
Cong Wang\IEEEauthorrefmark{5},
Yinqian Zhang\IEEEauthorrefmark{1}
}
\IEEEauthorblockA{
\IEEEauthorrefmark{1}Southern University of Science and Technology,
\IEEEauthorrefmark{2}University of British Columbia,
\IEEEauthorrefmark{3}East China Normal University,\\
\IEEEauthorrefmark{4}Texas A\&M University,
\IEEEauthorrefmark{5}City University of Hong Kong
}
}

\maketitle

\begin{abstract}
Trusted Execution Environments (TEEs) have enabled confidential Byzantine Fault-Tolerant (BFT) consensus systems with confidentiality and improved scalability. However, TEEs do not provide state continuity: during recovery, a compromised host can roll back a crashed enclave to a stale persistent state, significantly threatening both safety and availability. Existing defenses face a fundamental tradeoff: they either impose substantial overhead on critical consensus paths, reducing throughput and increasing latency, or incur prolonged recovery delays, hurting availability.

We present the first systematic taxonomy of rollback-resilient recovery for confidential BFT consensus, distilling prior approaches into four categories. We further expose their inherent limitations. Guided by this detailed analysis, we design \sysname, a protocol-aware recovery framework that breaks this tradeoff. Our key insight is that rollback protection in consensus systems should not be uniform. Different types of persistent states differ fundamentally in their state distributions, update behaviors, and representations. \sysname separates persistent state into metadata and logs according to these protocol-level properties and applies distinct recovery mechanisms to each type.
We formally model \sysname in Maude and verify its safety and liveness properties. 
We implement it on Braft and ZooKeeper using Intel TDX, and evaluate it in both LAN and WAN settings. Results show that \sysname achieves higher throughput, lower recovery latency, and better availability than state-of-the-art rollback-resilient baselines.
\end{abstract}

\IEEEpeerreviewmaketitle

\section{Introduction} \label{sec:intro}

Confidential Byzantine Fault Tolerant (BFT) consensus, which uses Trusted Execution Environments (TEEs) to improve system scalability and confidentiality, has recently gained significant traction. 
By leveraging TEEs, confidential BFT consensus enables a set of nodes to agree on an ever-growing, consistent log of transactions while preserving client confidentiality in the presence of Byzantine behavior (\ie, arbitrary protocol deviation) among nodes. 
Due to the promising consistency, confidentiality, and fault tolerance properties, it has been used in many decentralized or cloud services, including datastores~\cite{jeffery2024lskv, securekeeper, svr3}, blockchains~\cite{engraft, yan2020confidentiality, oasis2020}, cloud computing~\cite{2019ccf, 2023ccf, ccf2025}, and self-expiring data objects~\cite{gao2021teekap}. 

More specifically, confidential BFT consensus offers two key advantages over classic BFT protocols. 
First, due to the \textit{integrity} guarantees of TEEs, parties running inside TEEs cannot equivocate their messages (\eg, votes). Therefore, confidential BFT consensus achieves Byzantine fault tolerance by running Crash Fault Tolerant (CFT) protocols, such as Raft~\cite{raft2014}, inside TEEs. This design improves scalability by allowing a smaller node size and less computational overhead.
Second, due to the \textit{confidentiality} guarantees of TEEs, the secrecy of computations (\eg, blockchain transactions) is preserved against untrusted nodes, enabling wider applicability across sensitive use cases.

Nonetheless, a well-known Achilles' heel of confidential BFT lies in the lack of state continuity in TEEs. State continuity allows a TEE to preserve a consistent, tamper-resistant state across crashes and restarts. In fault-tolerant protocols, a persisted state enables a crashed node to recover and rejoin the system. However, an adversary that controls the host can provide stale state to the TEE during recovery, causing a rollback attack~\cite{engraft, 2023ccf}. 

Recent years have seen substantial efforts on achieving rollback resilience in such settings~\cite{Rote, narrator, Ariadne, Memoir, engraft, gupta2022dissecting, 2023rr, ccf2025, achilles, nimble}. We categorize these approaches into four categories: Trusted Counter (TC), Diskless Crash Recovery (DCR), Rollback Fault Tolerance (RFT), and Reconfiguration (RC). 
Yet each category exemplifies a fundamental trade-off between performance and availability: they either impose high overhead on critical consensus paths or incur prolonged recovery delays. (See \ssecref{sec:taxonomy} for a detailed discussion.)

In this paper, we aim to answer the question: \textit{How can we design rollback-resilient recovery for confidential BFT consensus that maintains high performance during normal operation and minimizes downtime during recovery?} 

Our key insight is that the above tradeoff stems from a common design choice: existing defenses protect persistent state uniformly—applying the same protection to all state—even though different types of state differ in how they are updated, stored, and recovered by the protocol.
By classifying persistent states according to their protocol-level characteristics, we can design recovery mechanisms tailored to each type.
For leader-based consensus, these states can be divided into two categories: metadata and logs. Metadata is a set of node-local control states that records each node’s view of the consensus process, while logs are replicated histories, with replications maintained according to consensus rules.

Guided by this distinction, we propose \sysname, a protocol-aware recovery framework for confidential BFT consensus that customizes recovery for metadata and logs.
{For metadata, \sysname uses a TC-based mechanism to securely and timely recover metadata from local storage. For logs, \sysname uses protocol-guided cluster recovery that leverages consensus replication rules to rebuild a safe log from other replicas.}
This design balances performance and availability: it confines TC overhead to infrequent metadata updates, avoiding TC updates on the high-frequency log replication path; it also reduces recovery downtime by avoiding epoch skips and allowing crashed leaders to resume without waiting for new leader election.

Applying such protocol-aware recovery, however, is nontrivial.
\emph{First}, for metadata, unexpected crashes may occur between the counter increment and metadata sealing. Without proper handling, a recovering node may fail to restore its metadata if the counter and persistent states are mismatched. Achilles~\cite{achilles} proposes skipping several epochs to avoid double voting, but this can render the node temporarily unavailable. We analyze the update process of metadata and design a binding update solution to address this issue (\ssecref{subsec:metadataDesign}).
\emph{Second}, in log recovery, a crashed leader may leave an entry replicated to only a subset of nodes, failing to reach quorum. To support leader recovery, the recovered leader must confirm its leadership and identify any unfinished log entries; otherwise, it may re-propose entries at the same indices, resulting in conflicting entries at the same log position (\ssecref{subsec:logDesign}).
\emph{Third}, log synchronization from the network is time-consuming. To mitigate this cost, \sysname uses metadata to identify the log range that needs to be synchronized. This approach reduces unnecessary log transmission during recovery and enables recovering nodes to rejoin the protocol with minimal delay (\ssecref{subsec:logDesign}).

We prototype \sysname by extending two open-source industrial CFT consensus platforms: Braft~\cite{braft}, a high-performance implementation of Raft~\cite{raft2014}, and ZooKeeper~\cite{ZooKeeper}, a widely adopted coordination service built on the Zab~\cite{zab} consensus protocol. We select Raft and Zab because they are the most widely used CFT consensus protocols for building confidential BFT consensus~\cite{svr3, 2019ccf, 2023ccf, engraft, jeffery2024lskv, securekeeper}. Both prototypes are implemented on top of Virtual Machine (VM)-based TEE, \ie, Intel TDX~\cite{tdx_white}. We formally model \sysname in Maude and verify its safety and liveness properties. We conduct extensive experiments on a public cloud platform to evaluate and compare \sysname with four counterparts, particularly their performance during recovery over LAN and WAN.

\bheading{Contributions.} Our main contributions are as follows:
\begin{packeditemize}
    \item We provide the first taxonomy of rollback-resilient solutions for TEEs in distributed systems, and organize them into four categories. We further analyze why applying uniform rollback protection is insufficient for persistent states with differing protocol-level characteristics in these systems.

    \item We propose \sysname, a protocol-aware recovery framework for confidential BFT consensus. It tailors recovery mechanisms to metadata and logs according to their characteristics. To ensure its correctness (\ie, safety and liveness), we perform a thorough security analysis and conduct a formal modeling and verification of \sysname.

    \item We implement \sysname atop Braft and ZooKeeper with Intel TDX, and integrate optimizations to reduce recovery communication. Our extensive evaluation demonstrates that \sysname delivers superior performance.
\end{packeditemize}

\section{Background} \label{sec:bg}
\subsection{TEEs and Their State Continuity}

Trusted Execution Environments (TEEs) are hardware-supported execution environments that protect sensitive code and data. They achieve this protection through mechanisms such as memory encryption, hardware-enforced isolation, and remote attestation. Representative platforms include enclave-based TEEs such as Intel SGX~\cite{sgx2013}, VM-based TEEs such as Intel TDX~\cite{tdx_white} and AMD SEV~\cite{sev2020strengthening}, and other hardware isolation architectures such as ARM TrustZone~\cite{Alves2004TrustZone}. Compared with enclave-based TEEs, VM-based TEEs allow existing applications to benefit from TEE protection with minimal code changes and low additional overhead. We therefore build our implementation on VM-based TEEs. 

TEEs have been adopted in a wide range of stateful applications, including blockchains~\cite{ekiden, Teechain, yan2020confidentiality, TeeRollup, mecury, fides}, trusted storage~\cite{oprea2007integrity,jeffery2024lskv, securekeeper, svr3}, authentication rate limiting~\cite{Ariadne}, and cloud computing services~\cite{2019ccf, 2023ccf, ccf2025, nimble, narrator}. However, despite their popularity, TEEs still lack state continuity guarantees: rollback attacks can revert a TEE to a prior state, undermining the security of applications built atop them.

\bheading{State Continuity.} State continuity of TEEs mandates that when a stateful TEE application resumes execution from an interruption (\eg, reboots or system crashes), it must resume in the same state as before ~\cite{Memoir}. This property is critical for applying TEEs in distributed systems, since nodes' ability to recover from crashes directly determines system reliability.~\cite{ghemawat2003google, hamilton2007designing, Grapevine}. Existing TEE platforms provide sealing functionality, allowing TEE applications to encrypt and store their state to untrusted persistent storage. The sealing key can be configured to be accessible to all enclaves with the same MRENCLAVE or MRSIGNER~\cite{narrator}.

\bheading{Rollback Attacks.} The sealing functionality can ensure the integrity of retrieved data, but does not provide freshness guarantees~\cite{priebe2018enclavedb}. 
An adversary controlling the OS can roll back a TEE application to a previous state by providing it with stale data, resulting in a rollback attack~\cite{wilde2024forking}.
These attacks break state continuity and undermine the security guarantees of TEEs for various stateful applications, particularly confidential BFT consensus systems.
For instance, if a node is rolled back after casting a vote, it may re-enter a previous state and cast the vote again. Such repeated voting can lead to equivocation, violating the safety properties of the protocol.

\subsection{Confidential BFT Consensus} \label{subsec:leaderCFT} 
Confidential BFT consensus ports CFT consensus protocols into TEEs to provide BFT services. 
Examples include SVR3~\cite{svr3} adopted by Signal~\cite{signal}, CCF~\cite{2023ccf} deployed on the Azure cloud platform~\cite{ccfazure}, Engraft~\cite{engraft}, Hyperledger Fabric~\cite{HyperSrds}, and SecureKeeper~\cite{securekeeper}. 
Among them, SVR3, CCF, and Engraft are built on Raft~\cite{raft2014}, whereas SecureKeeper uses Zab~\cite{zab}.
Raft and Zab are well-established leader-based CFT consensus protocols. Both provide similar protocol-level semantics, which are leveraged in our design.

\begin{packeditemize}
    \item\textbf{Epoch.} The leader-based consensus protocol partitions time into logical units called epochs (referred to as terms in Raft). For each epoch, at most one node, called the leader, is elected and agreed upon by a quorum. When the leadership needs to change, the epoch is incremented to reflect the transition.

    \item\textbf{Leader Election.} When nodes detect that the leader is unavailable, they initiate a leader election to preserve liveness. Specifically, each node casts at most one vote per epoch and persists this voting behavior to prevent double voting. A candidate node becomes a leader when it receives votes from the majority of nodes. A non-leader node is called a follower. In this paper, we refer to the persisted epoch and voting behavior as metadata.
 
    \item\textbf{Log.} The log is a sequence of client transactions maintained by nodes, where each entry contains one or a batch of transactions. Each entry is uniquely identified by an epoch number and a monotonically increasing index. Nodes commit entries in a contiguous sequence sorted by index without gaps. Unlike metadata, logs are replicated across nodes following consensus rules.
\end{packeditemize}

\section{Recovery Taxonomy} \label{sec:taxonomy}
We present the first systematic analysis of TEE recovery mechanisms in distributed systems, categorizing them into four groups, as shown in Table~\ref{table:recoveryComparison}.
This taxonomy provides a direct comparison of the approaches, as we present next. More importantly, it shows why rollback recovery for confidential BFT must be protocol-aware rather than uniform.

\bheading{Trusted Counter (TC).} TC is a tamper-resistant counter whose value, once incremented, cannot be reverted to a previous value~\cite{VirtualCounter}. It can be implemented in hardware (\eg, TPM~\cite{VirtualCounter} or SGX monotonic counters~\cite{sgx_counter}) or in software (\eg, distributed trusted KV~\cite{Rote, narrator}) (Appendix~\ref{appen:TC}). 
It is the most popular and general approach for addressing TEEs' rollback attacks~\cite{Rote, narrator, narrator-pro, Ariadne, Memoir, engraft}. 

This approach performs two operations for each state update: (1) incrementing the counter value, and (2) sealing the updated state together with the counter value in persistent storage. During recovery, a node retrieves the sealed states and validates its freshness by comparing the stored counter value against the current one. Although this design is relatively simple and allows the system to recover the exact pre-crash state, it suffers from the following limitations:

\begin{packeditemize}
    \item \textit{High read/write overhead.} Each state update requires a counter write, introducing substantial overhead. Updating a hardware-backed counter, such as a TPM counter, takes roughly 97 ms, while reading it for state verification takes about 35 ms~\cite{Ariadne}. Software-based counters typically incur one or two extra communication rounds~\cite{Rote, narrator, narrator-pro}. Such overhead makes TC unsuitable for a frequently updated state.
    
    \item \textit{Inc-store consistency dilemma.} {Since counter increment and state sealing cannot be performed atomically, two usage patterns have emerged. The \textit{inc-then-store} pattern~\cite{Rote, narrator, narrator-pro, engraft} increments the counter before sealing the state, preserving safety but risking availability: a crash after the increment but before sealing makes recovery impossible. In contrast, the \textit{store-then-inc} pattern~\cite{Memoir, Ariadne} favors availability over safety, since a crash during the interval can result in rollback. Further discussion is provided in Appendix~\ref{appen:TC}.}

    \item \textit{Detection-only recovery.} {TC can determine the freshness of a persisted state by comparing counter values, but it cannot identify which state is safe for recovery; in other words, it provides only rollback detection, not full recovery.}
\end{packeditemize}

\begin{table}[t]
    \caption{Rollback-resilient solutions.}
    \label{table:recoveryComparison}
    \centering
    \scalebox{0.97}{
    \begin{tabular}{@{}lcccc@{}}
        \toprule[1pt]
        Approach & \tabincell{c}{Recovery \\ Pattern} & \tabincell{c}{Recovery \\ Completeness} & \tabincell{c}{Performance \\ Overhead} & \tabincell{c}{Availability \\ Degradation} \\
        \midrule
        TC   & Local       & High       & High & Low \\
        DCR  & Distributed & Medium     & Low  & High \\
        RFT  & Local       & High$^{*}$ & High & Low \\
        RC   & Distributed & Low        & Low  & High \\ 
        \hline
        \textbf{\sysname} & Hybrid & High & Low & Low \\
        \bottomrule[1pt]
    \end{tabular}}
    
    \vspace{3pt}
    {\footnotesize $^{*}$ RFT enables a node to recover its state, but cannot prevent state rollback.}
    \vspace{-4mm}
\end{table}

\bheading{Diskless Crash Recovery (DCR).} DCR enables a node to recover its state from the in-memory state of other nodes in the cluster~\cite{michael2016providing, liskov12vr, chandra2007paxos, konczak2011jpaxos}. Specifically, DCR is effective at recovering redundant data, such as logs, because the protocol's replication rules allow recovery to locate a node with a sufficiently safe copy. It relies only on naive log replication and introduces no additional overhead on the consensus path. However, directly applying DCR to confidential BFT recovery exposes the following shortcomings.

\begin{packeditemize}
    \item \textit{Epoch skipping for safety.} DCR cannot recover node-local metadata such as epochs and voting behavior. To prevent rollback-induced double voting, a recovering node must first infer a safe upper-bound epoch based on the protocol's metadata update rules. The node then recovers directly into that epoch, skipping earlier epochs. As a result, it may remain unavailable until the cluster reaches the same epoch.

    \item \textit{Protocol-specific recovery rules.} The safe epoch inferred by DCR depends on the protocol's metadata update rules. Different protocols may require skipping varying numbers of epochs, and those without bounded epoch growth may fail to provide a usable safe epoch. Consequently, both the recovery rules and the unavailability period are protocol-specific.

    \item \textit{No leader recovery.} DCR requires an active leader. If the leader crashes, recovering nodes can only resume as followers and must wait for a new leader to be elected. This dependency prolongs the period of unavailability after leader failures.
\end{packeditemize}

\bheading{Rollback Fault Tolerance (RFT).}  
RFT addresses rollback attacks by adjusting read/write quorum sizes according to the number of potentially rolled-back nodes. RR~\cite{2023rr} modifies read/write quorum sizes to guarantee intersection with up-to-date nodes, whereas FlexiBFT~\cite{gupta2022dissecting} increases the overall node size to $3f+1$, aligning with classical BFT requirements. In RFT, a recovering node can quickly retrieve its state—which is not protected against rollbacks—and rejoin the protocol without affecting system availability. This design choice brings the following limitations:

\begin{packeditemize}
    \item \textit{Scalability limitation.} 
    Increasing node and quorum sizes undermines log commit performance, particularly in large deployments. Evaluation results show that when $f$ is small, the performance gap between the $2f{+}1$ and $3f{+}1$ configurations atop Braft (\ie, Braft-DR and Braft-RFT, respectively) is around 10\%. However, this gap grows with larger $f = 20$, reaching roughly 20\% (\ssecref{sec:normalEvaluation}).
    Increasing the node size also conflicts with prior TEE-aided designs that aim for smaller node sizes~\cite {bessani2023vivisecting, Clement}.
\end{packeditemize}

\bheading{Reconfiguration (RC).} RC enables the system to dynamically modify its node set by adding or removing nodes. It is a more general approach than simple recovery, as a recovering node can rejoin the system as a newly added participant without having to recover any prior state. 
Notable confidential BFT protocols, CCF~\cite{2023ccf} and Recipe~\cite{2025recipe} adopt this design. Since RC does not affect the log replication, no extra overhead is introduced.  Its drawbacks appear during recovery:
 
\begin{packeditemize}
    \item \textit{High-cost membership changes.} RC typically requires complex recovery designs. To ensure safety, most protocols require nodes to reach consensus on the order of membership changes, so that all nodes apply them consistently. This coordination can delay concurrent log commits. For example, RC in Raft requires two rounds of consensus to perform a membership change safely.

    \item \textit{Expensive synchronization.} A rejoining node must synchronize the entire application state from scratch~\cite{2023ccf}. For applications with large state, such as blockchains, this process can involve transferring hundreds of gigabytes of data and may take several hours to days~\cite{kim2021ethanos,feng2024slimarchive}.
\end{packeditemize}

\bheading{Summary.} Existing recovery approaches in confidential BFT consensus follow a one-size-fits-all design: each approach tries to protect or recover all persistent states with the same mechanism. As a result, the system pays the worst-case cost for the least suitable state type, rather than exploiting protocol-level knowledge and the characteristics of metadata and logs. This observation motivates a protocol-aware recovery design that customizes and strengthens recovery for each state type.

\section{Problem Statement} \label{sec:problemStatement}

\subsection{System Model} \label{subsec:sysmodel}
Following the model of prior confidential BFT consensus~\cite{engraft, 2023ccf}, we consider a distributed system maintained by $n = 2f+1$ nodes $\{p_1, p_2,\ ...\ , p_{n}\}$, each equipped with a TEE. Confidential BFT consensus runs entirely inside each node's TEE.
We assume a Public Key Infrastructure (PKI): each node $p_i$ has a public/private key pair, denoted by $(pk_i, sk_i)$, in which the private key is accessible only within the node's TEE. A message $m$ signed with $sk_i$ is denoted by $m_{\sigma_i}$. 
We assume that a finite set of clients sends transactions to nodes' TEEs for confidential BFT service over encrypted and authenticated channels (\eg, TLS). 

\bheading{Threat model.}
We assume an adversary $\mathcal{A}$ that can corrupt at most $f$ nodes at any time and any number of clients. Following prior study~\cite{engraft, 2023ccf}, corrupted nodes are \textit{Byzantine}, \ie, behaving arbitrarily, with the exception that TEE integrity and confidentiality cannot be breached (introduced shortly). The other nodes that faithfully follow the protocol and remain operational (\ie, participating in consensus) are \textit{correct} nodes. The rationale behind this assumption is provided in Appendix~\ref{appen:restriction}.

The adversary gains full control over the OS of the corrupted node: it can manipulate network messages between TEEs and arbitrarily start, stop, and invoke TEEs.
Moreover, the adversary can provide the TEE with stale data to rollback the state~\cite{Rote, Ariadne, Memoir, Dijk07, ICE, engraft}.  
We do not consider cloning attacks~\cite{Rote, narrator} or micro-architectural side-channel attacks~\cite{Schwarz:2019:zombieload,van:2019:ridl,chen:2019:sgxpectre}, as they are orthogonal to this work. 
Cloning attacks can be mitigated using TPM PCR~\cite{Ariadne}, while side-channel attacks can be addressed via software-level countermeasures, particularly in cryptographic libraries such as OpenSSL and Intel SGX SSL~\cite{shih2017tsgx, oleksenko2018varys}.

\bheading{Network Model.} 
We adopt the partially synchronous network model~\cite{dwork1988consensus}, which is commonly used in consensus~\cite{castro1999practical, HotStuffYin2019, fastHotStuff, gai2023scaling}. 
In this model, there is an established bound $\Delta$ and an undefined Global Stabilization Time (\textsf{GST}). After the \textsf{GST} point, the delivery of any message transmitted between two correct nodes within the $\Delta$ limit is guaranteed. That is, the system behaves \textsl{synchronously} following the \textsf{GST}. Liveness is guaranteed after \textsf{GST}.

\subsection{Problem Statement}\label{subsec:problemstatement}
In confidential BFT consensus, each node runs a customized leader-based CFT protocol (Appendix~\ref{appen:customization}) inside its TEE to commit and execute client transactions. The protocol proceeds in a sequence of epochs, each representing a logical leader term. Ideally (\ie, without rollback issues), one node is elected as the leader in each epoch, and a majority of nodes agree on this choice (\ie, election safety).
The leader batches client transactions, proposes them as log entries, and replicates each entry to the logs of all nodes $p_i$ within the TEE. Each log entry is tagged with an epoch and log index $(ep, idx)$. Here, $ep$ denotes the epoch in which the entry is proposed, and $idx$ denotes its position in the log. 
The leader appends each entry to its local log and persists it. It then replicates the entry to all followers. Each follower appends and persists the entry before sending an acknowledgment. Upon receiving acknowledgments from a majority of followers, the leader marks the entry as committed.

As in prior work~\cite{castro1999pbft}, confidential BFT consensus provides two fundamental guarantees: \textbf{safety} and \textbf{liveness}. Safety requires the service to be linearizable: no two correct nodes commit different entries at the same log index $idx$. Liveness ensures that every transaction submitted by a client is eventually committed.

\bheading{Recovery under TEE Rollbacks.}
The integrity guarantees of TEEs allow confidential BFT consensus to maintain safety and liveness during normal execution~\cite{castro1999pbft}. However, recovery in the presence of TEE rollbacks introduces additional challenges. In leader-based consensus protocols such as Raft~\cite{raft2014}, safety and liveness can be refined into three critical conditions during recovery. Specifically, to preserve safety, the recovery procedure must satisfy the following two properties:

\begin{definition}[Election Safety]
For any epoch $ep$ and any two distinct nodes $p_i$ and $p_j$, it is impossible for both $p_i$ and $p_j$ to be elected leaders in $ep$.
\end{definition}

\begin{definition}[Leader Completeness]
For any log entry $en$ committed at index $idx$ in epoch $ep$, every leader elected in a later epoch $ep' > ep$ must contain $en$ at index $idx$ in its log.
\end{definition}

Election safety guarantees a unique leader for each epoch, while leader completeness ensures that all committed log entries are preserved across subsequent leaders. During recovery, the system must maintain correct commitment progress and restore all committed entries.

For liveness, the recovery procedure must ensure that an uncorrupted recovering node can eventually resume participation, as formalized below.

\begin{definition}[Recovery Liveness]
For any uncorrupted node $p_i$, if $p_i$ starts recovery and remains uncorrupted, then $p_i$ eventually completes recovery and resumes participation in the protocol.
\end{definition}

\section{\sysname Overview} \label{sec:overview}
We present an overview of \sysname, a protocol-aware recovery framework for confidential BFT consensus. \sysname first characterizes the persistent state maintained by leader-based consensus protocols (\ssecref{subsec:character}). It then applies a recovery strategy specialized for each state type (\ssecref{subsec:protocolaware}). 
This framework can achieve rollback-resilient recovery with both high performance and high availability, while preserving \textit{election safety}, \textit{leader completeness}, and \textit{recovery liveness}.

\subsection{Characterizing System State} \label{subsec:character}
In a leader-based CFT consensus, three types of system state require external persistence: metadata, logs, and snapshots. Snapshots are primarily used to accelerate log synchronization; therefore, their recovery is not essential for preserving safety guarantees and does not need rollback-resilient mechanisms. In this work, we focus on the two safety-critical states—metadata and logs—as summarized in Table~\ref{table:character}.

\bheading{Metadata.}
Metadata denotes a set of safety-critical control variables that govern leader election and epoch transitions. Metadata is updated infrequently and generally changes during leader transitions. It typically includes the current epoch and the node's voting record within that epoch. For example, in Raft, metadata consists of \currentTerm (the epoch identifier) and \votedFor (the candidate voted for in that epoch). The epoch identifier allows nodes to recognize larger epochs and reject stale leadership attempts. The voting record prevents double voting, ensuring election safety.

Two properties make metadata distinct for recovery. First, metadata is node-local: it records a node’s own view of consensus. As a result, other replicas cannot reliably reconstruct the recovering node's exact epoch or vote. Therefore, recovery must balance safety and availability: ensuring freshness of metadata preserves election safety, but entering a higher epoch may temporarily reduce availability. This trade-off makes metadata well-suited for TC-based local freshness protection, which enables precise, node-local recovery. Second, metadata is scalar and can be stored in a register-like form. In particular, the epoch increases monotonically during leader transitions, which aligns naturally with the TC design. Moreover, the low update frequency of metadata ensures that introducing TC incurs minimal additional overhead.

\bheading{Log.}
Each node maintains a log that consists of ordered entries. An entry is considered committed once it has been replicated to a majority of nodes. Leader completeness ensures that all committed entries are preserved across subsequent leaders. To recover safely, we only need to make sure that the commit process is correct and that committed entries are durable. Each node that replicates an entry must ensure the entry's durability before it is considered committed. By doing so, all committed entries remain available for future recovery and are preserved by subsequent leaders.

Logs differ from metadata in distribution and form. Unlike metadata, which is node-local, logs are replicated across multiple nodes in a structured way. Therefore, we can leverage the replicas to reconstruct a safe log during recovery. Moreover, logs are large and stored on disk, making it infeasible to maintain a direct mapping to a trusted counter. Version numbers can be used to detect inconsistencies, but they do not enable full recovery of log contents. 

\begin{figure}[t]
  \centering
  \includegraphics[width=0.95\linewidth]{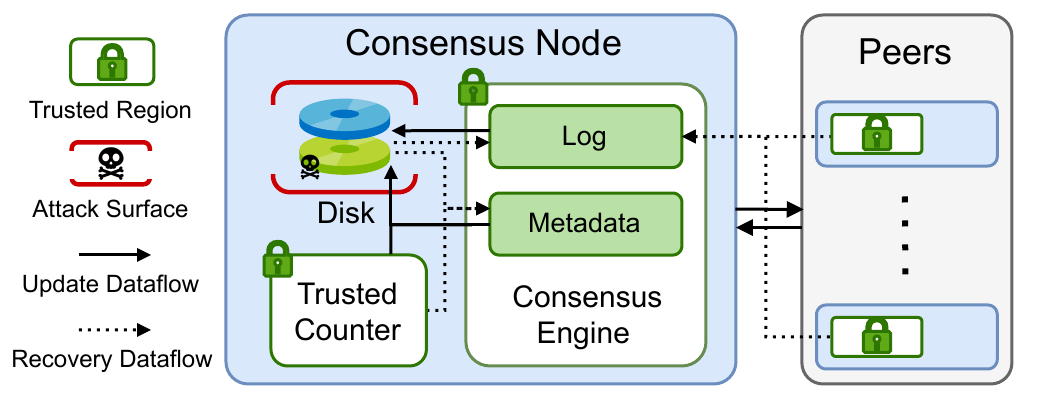}
  \caption{Architecture of \sysname.}
  \label{fig:ach}
  \vspace{-4mm}
\end{figure}

\subsection{Protocol-Aware State Recovery} \label{subsec:protocolaware}
We leverage protocol-level knowledge, \ie, the inherent characteristics of logs and metadata, to carefully tailor recovery strategies, as illustrated in \figref{fig:ach}.

\bheading{Metadata Recovery.}
\sysname protects metadata using a trusted counter (TC), leveraging the counter's node-local nature to enable precise recovery. The TC enables a recovering node to verify the freshness of sealed metadata locally, ensuring election safety without coordinating with other replicas. This approach allows precise local recovery while avoiding unnecessary loss of availability, capturing the trade-off between freshness guarantees and state continuity.

Raw TC alone, however, is insufficient: counter increments and metadata sealing are not atomic, so a mismatch may indicate either a rollback or a crash during the update window. To address this, \sysname proposes a binding-update design, motivated by the similarity between metadata and the TC. The binding-update design ties the epoch to the counter value to ensure consistent metadata recovery. When the counter and epoch match, the node can safely restore the sealed metadata. Otherwise, the counter serves as a safe epoch anchor, enabling recovery without conservative epoch skipping.

\bheading{Log Recovery.} 
For logs, \sysname relies on protocol-guided cluster recovery rather than per-entry TC protection. Consensus replication ensures that committed entries are preserved on a quorum of replicas. During recovery, a node can query a quorum of replicas to reconstruct a log that safely includes all committed entries. To guarantee safety, each node that replicates an entry must ensure the entry's durability before it is considered committed.

\sysname also tailors recovery to the node’s role to improve availability. A recovering follower can use metadata to avoid unnecessary log transfer during catch-up. A recovering leader, in contrast, can resume service after reconstructing a safe log without waiting for a new leader election. By leveraging both protocol-level guarantees and node role information, \sysname enables safe and efficient log recovery while preserving availability.

\begin{table}[t]
    \caption{The characteristics of metadata, log, and snapshot.} 
    \label{table:character}
    \centering
   \scalebox{1.05}{
    \begin{tabular}{@{}lccccc@{}}
        \toprule[1pt]
        State Type & \tabincell{c}{Safety \\ Critical} & \tabincell{c}{Recovery \\ Requirement } & \tabincell{c}{System \\ Redundancy} & \tabincell{c}{Update \\ Frequency} \\
        \midrule
        Metadata & \cmark & Precise & \xmark & Low \\
        Log & \cmark & Loose & \cmark & High \\
        \rowcolor[gray]{0.9}  
        Snapshot & \xmark & Loose & \xmark & Low \\
        \bottomrule[1pt]
    \end{tabular}}
    \vspace{-4mm}
\end{table}

\section{\sysname Design}

\subsection{Data Structures and Interfaces}
\bheading{Data Structures.} \sysname handles the update and recovery of metadata and logs separately.

\iheading{1) Metadata.}
Metadata captures epoch and leader-election information, such as the current epoch $ep$ and the node's voting record. We denote the metadata of node $p_i$ as $md_i$. Each metadata update corresponds to an operation $op$ (\eg, sending an election vote), which is executed only after the updated metadata has been sealed.

\iheading{2) Log.}
Each node $p_i$ maintains a local log $log_i$, which consists of a sequence of entries $en_j$ indexed by $j$. A log entry is identified by its epoch and index; we refer to this pair $(ep,j)$ as the entry's \emph{epoch-index tag}. We compare epoch-index tags lexicographically: $(ep_1, idx_1) >
(ep_2, idx_2)$ iff $ep_1 > ep_2$, or $ep_1 = ep_2$ and $idx_1 > idx_2$.

\bheading{Interfaces.} \sysname relies on sealing interfaces for protecting states stored in external storage, and on TC interfaces for maintaining freshness across crashes.

\iheading{1) Sealing interfaces.} TEEs provide two sealing interfaces: 

\begin{packeditemize}
    \item \Seal (\textit{data, h}):
    Encrypt \textit{data} with the TEE-internal sealing key and then store the associated encrypted data at the location identified by handle $h$ in external storage. 
    \item $data \leftarrow$ \UnSeal(\textit{h}):
    Retrieve the sealed data from an external storage with the handle $h$, decrypt it, and return the associated plaintext \textit{data}.
\end{packeditemize}

\iheading{2) TC interfaces.} We use two TC operations:

\begin{packeditemize}
    \item $ctr \leftarrow$ \IncTC($k$): Atomically increment the counter identified by key $k$ by one and return the new value $ctr$.
    \item $ctr \leftarrow$ \ReadTC($k$): Return the current value $ctr$ of the counter identified by key $k$ without modifying it.
\end{packeditemize}

Each node maintains two trusted counters for different recovery purposes:
\begin{packeditemize}
    \item $TC_{md}$ records metadata updates. Its value is bound to the epoch in sealed metadata, enabling rollback detection and safe recovery when the sealed states and the counter mismatch.
    \item $TC_{\mathsf{role}}$ records a node's role within the current epoch. It is initialized to 0 and incremented whenever the node changes its role between follower and leader. Thus, when the counter value is odd, the node is a leader; when it is even, the node is a follower.
\end{packeditemize}

\bheading{Initialization.}
Each node boots from a hard-coded genesis configuration and initializes with default metadata (\ie, $ep = 1$) and an empty log. The configuration information, including key pairs, is stored on local disks in an encrypted and authenticated form.

In this work, we consider a static configuration: a recovering node can retrieve the configuration to obtain its own key pairs and the public keys of other nodes for communication. Dynamic reconfiguration is discussed in \ssecref{sec:conclusion}.

The trusted counters are initialized to match the starting system state: $TC_{md}$ is set to 1 to align with the initial epoch, while $TC_{role}$ is set to 0 to represent the initial follower role. These initializations ensure that metadata freshness and node role tracking start from a consistent, well-defined state.

\subsection{Metadata Recovery Mechanism} \label{subsec:metadataDesign}

\subsubsection{Metadata Update} 
\sysname protects metadata with TC, since metadata is node-local and requires more precise recovery to gain higher availability.
For each metadata update, the node first increments $TC_{md}$ and then seals the metadata together with the counter value, following an \textit{inc-then-store} scheme~\cite{Ariadne} (Fig.~\ref{fig:metadata}a).

Using a TC provides a local freshness check for metadata, but a raw TC is insufficient: if a crash occurs after the counter is incremented but before the updated metadata is sealed, recovery observes a counter mismatch and cannot distinguish between a benign crash and a rollback. \sysname resolves this ambiguity by binding the scalar epoch to $TC_{md}$. The key rule is to advance both values in lockstep: every metadata update increases the epoch by exactly one and performs one $IncTC(md)$. This design turns $TC_{md}$ into a record of epoch transitions. When the counter value matches the sealed metadata, the node can safely restore the full metadata from persistent storage. In the event of a mismatch, \sysname uses the counter value as the recovered epoch. To maintain safety, the node disables voting in this epoch due to uncertainty about its previous votes.

When a node receives a request with a higher target epoch $ep^*$ than its current epoch $ep$, it executes the following loop until it reaches $ep^*$:

\noindent \whiteding{1} Invoke \IncTC(md) to advance the trusted counter $TC_{md}$ by one step and obtain the new value $ctr$.

\noindent \whiteding{2} Update the metadata: increment $md.\textit{epoch}$ by 1 and set $md.\textit{vote}$ according to the request type. For all intermediate epochs, set $md.\textit{vote}$ to non-voting. When $md.\textit{epoch}$ reaches the target $ep^*$, set $md.\textit{vote}$ to the selected candidate if the request is a vote request; otherwise, it remains non-voting.

\noindent \whiteding{3} Persist metadata: call \Seal($md | ctr$, $h_{md}$) to store the updated metadata along with the counter value.

\noindent \whiteding{4} Send vote reply: for vote requests, the node sends the reply when $md.\textit{epoch}$ reaches $ep^*$. The reply is sent after the corresponding metadata has been sealed to ensure consistency.

A multi-epoch jump is executed as a sequence of single-epoch updates. All intermediate epochs are marked non-voting. The final epoch records a candidate only if the triggering request is a granted vote. This binding update keeps $md.\textit{epoch}$ synchronized with $TC_{md}$ and prevents the node from casting a different vote in any previous epoch.

\subsubsection{Metadata Recovery}
As shown in Fig.~\ref{fig:metadata}b, metadata recovery follows the update scheme. The node first checks the freshness of the sealed metadata using $TC_{md}$. If the sealed counter value matches the current counter value, the node recovers from the sealed state. Otherwise, it recovers to a safe epoch without risking double voting.
The binding-update design binds $TC_{md}$ to the current epoch, maintaining a correspondence between the counter value and epoch. This correspondence provides a reliable anchor for safe recovery. The recovery process then proceeds as follows:

\noindent \whiteding{1} Load sealed metadata: the node invokes \UnSeal($h_{md}$) to retrieve $md^{\prime}$ together with the sealed counter value $ctr^{\prime}$.

\noindent \whiteding{2} Read current counter: the node calls \ReadTC(md) to obtain the current counter value $ctr$ for freshness verification.

\noindent \whiteding{3} Check counter match: if $ctr = ctr^{\prime}$, the sealed metadata is fresh. The node sets $md \leftarrow md^{\prime}$ and finishes recovery.

\noindent \whiteding{4} Handle mismatch: if $ctr \neq ctr^{\prime}$, the sealed metadata may be outdated. Then the node starts metadata repair with $ctr$.

\noindent \whiteding{5} Repair metadata: the node reconstructs metadata by setting $md.\textit{epoch} \leftarrow ctr$ and $md.\textit{vote}$ to non-voting. The node then resumes execution using this repaired metadata.

\begin{figure}[t]
    \centering
    \includegraphics[width=0.95\linewidth]{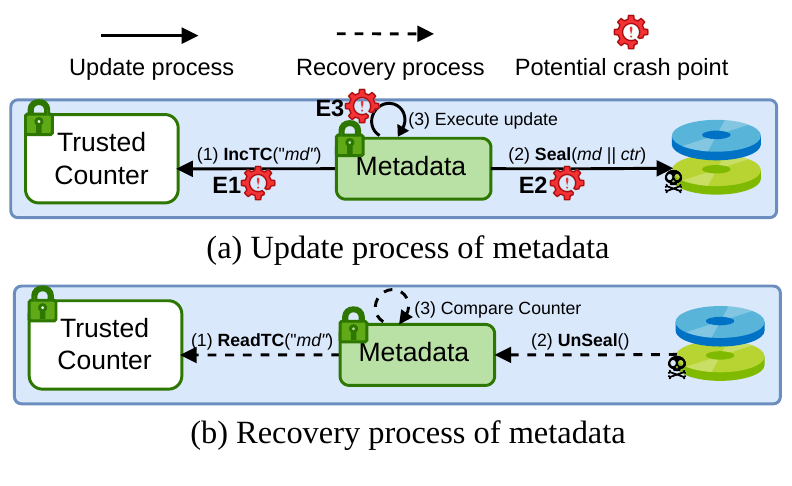}
    \caption{Recovery design of metadata.}
    \label{fig:metadata}
    \vspace{-4mm}
\end{figure}

\subsubsection{Crash Window Analysis}
The inc-then-store update is not atomic. Crashes can occur in three intervals, as illustrated in Fig.~\ref{fig:metadata}a: before incrementing $TC_{md}$ (\textbf{E1}); after incrementing $TC_{md}$ but before sealing the metadata (\textbf{E2}); and after sealing the metadata but before sending the vote reply or completing the triggering operation (\textbf{E3}).
Our analysis shows that each crash falls into one of two cases. In the match case, the sealed metadata can be safely restored. In the mismatch case, the node repairs to a non-voting epoch.

\begin{packeditemize}
\item \textbf{E1 Crash.}  
At \textbf{E1}, the node has prepared an update, but neither the epoch nor the $TC_{md}$ has changed. Since nothing has been written to persistent metadata, there is no mismatch. Recovery can safely restore the previous metadata, and no actions or effects have been produced.

\item \textbf{E2 Crash.}
At \textbf{E2}, $TC_{md}$ has advanced by one epoch, but the corresponding metadata has not yet been sealed. This creates a mismatch between the counter and the sealed metadata. This mismatch is indistinguishable from an actual rollback, rendering detection unreliable. The binding update design resolves this issue by using the counter value as the recovered epoch. Since no sealed metadata is available, the node must disable voting in that epoch to prevent casting or resending votes with incomplete metadata.

\item \textbf{E3 Crash.}
At \textbf{E3}, both $TC_{md}$ and the sealed metadata have been updated.
If no rollback occurred, the counter and sealed metadata match, allowing the node to safely restore the sealed metadata. With the complete metadata, any previously recorded voting behavior can also be safely replayed.
If a rollback occurred, a mismatch between the counter and the metadata will be observed. Recovery then uses the counter value to restore the epoch and marks the node as non-voting.
\end{packeditemize}

\subsection{Log Recovery Mechanism} \label{subsec:logDesign}
A recovering node begins log recovery once metadata recovery has established a safe epoch. Unlike metadata, logs are replicated according to consensus rules, so \sysname can recover them from other replicas. The role counter $TC_{role}$ allows a recovering node to determine its previous role before the crash, and \sysname uses different recovery paths for the two cases.
Alg.~\ref{algo:unified-log-recovery} summarizes both paths.

\begin{algorithm}[t]
\small
\caption{Log recovery at node $p_i$}
\label{algo:unified-log-recovery}
\begin{algorithmic}[1]
\STATE $resp \gets 0$ \COMMENT{the number of \msg{ReplyRecover} messages received}
\STATE $tag_{\max} \gets \bot$ \COMMENT{the most up-to-date last-entry tag observed}
\STATE $src \gets \bot$ \COMMENT{the ID of the node selected as the log source}
\STATE $done \gets \textbf{false}$ \COMMENT{whether follower recovery has completed}
\STATE
\STATE \textbf{upon} $\langle \msg{LogRecover}, h_{log} \rangle$:
\begin{ALC@g}
    \STATE $log_i \gets \UnSeal(h_{log})$
    \STATE $ep_i \gets$ recovered metadata epoch of $p_i$
    \STATE $last_i \gets$ last-entry tag of $log_i$
    \STATE $roleCtr_i \gets \ReadTC(role)$
    \IF{$roleCtr_i$ is even}
        \STATE broadcast $\langle \msg{FollowerRecover}, i, ep_i, last_i, non \rangle$
    \ELSE
        \STATE $ep_i \gets ep_i+1$
        \STATE broadcast $\langle \msg{LeaderRecover}, i, ep_i, last_i, non \rangle$
    \ENDIF
\end{ALC@g}
\STATE
\STATE \textbf{upon receiving} $\langle \msg{FollowerRecover}, j, ep_j, last_j, non \rangle$:
\begin{ALC@g}
    \STATE \textbf{if} $role_i \neq \text{\leader}$ or $ep_i < ep_j$ \textbf{then} \textbf{return} \textbf{end if}
    \STATE $\textit{suffix} \gets$ log entries after $last_j$
    \STATE send $\langle \msg{ReplyLog}, ep_i, \textit{suffix}, i, non \rangle$ to $p_j$
\end{ALC@g}
\STATE
\STATE \textbf{upon receiving} $\langle \msg{ReplyLog}, ep, \textit{suffix}, j, non \rangle$:
\begin{ALC@g}
    \STATE \textbf{if} $non$ is invalid or $done$ or $ep < ep_i$ \textbf{then} \textbf{return} \textbf{end if}
    \STATE append valid $\textit{suffix}$ to $log_i$
    \STATE $done \gets \textbf{true}$
\end{ALC@g}
\STATE
\STATE \textbf{upon receiving} $\langle \msg{LeaderRecover}, j, ep_j, last_j, non \rangle$:
\begin{ALC@g}
    \IF{$ep_i > ep_j$ or ($ep_i = ep_j$ and $vote_i=\textsf{disabled}$)}
        \STATE \textbf{return}
    \ENDIF
    \IF{$ep_i \leq ep_j$}
        \STATE \COMMENT{use the binding update of metadata}
        \STATE $ep_i \gets ep_j$
        \STATE $vote_i \gets$ \textsf disabled 
    \ENDIF
    \STATE $last_i \gets$ last-entry tag of $log_i$
    \STATE send $\langle \msg{ReplyRecover}, ep_i, last_i, i, non \rangle$ to $p_j$
\end{ALC@g}
\STATE
\STATE \textbf{upon receiving} $\langle \msg{ReplyRecover}, ep, last, j, non \rangle$:
\begin{ALC@g}
    \STATE \textbf{if} $non$ is invalid \textbf{then} \textbf{return} \textbf{end if}
    \STATE $resp \gets resp + 1$
    \IF{$last$ is more up-to-date than $tag_{\max}$}
        \STATE $tag_{\max} \gets last$; $src \gets j$
    \ENDIF
    \IF{$resp = f+1$}
        \STATE synchronize log with $p_{src}$ and resume proposing
    \ENDIF
\end{ALC@g}
\end{algorithmic}
\end{algorithm}

\subsubsection{Log Recovery for Followers} 
To ensure correctness, the commit process must be properly executed: every entry considered committed must be durably stored on a quorum of replicas. A recovering follower, therefore, needs to catch up with a safe leader. Generic DCR-style recovery may collect $f{+}1$ replies to identify the latest leader and ensure that entries still being replicated are not lost. In \sysname, metadata recovery already restores a safe epoch for the follower. Therefore, a recovering follower can synchronize its log from any leader whose epoch is no lower than its own.

The follower recovery path leverages this observation to reduce message overhead. First, the follower unseals its local log from persistent storage and determines its last-entry tag. It then broadcasts a \msg{FollowerRecover} message containing its recovered epoch and last-entry tag (Alg.~\ref{algo:unified-log-recovery}, L6–L12). Only the leader whose epoch is no lower than the follower’s epoch responds. The reply contains the log entries after the follower’s last-entry tag. This allows the follower to reuse its locally persisted entries and fetch only the missing suffix (Alg.~\ref{algo:unified-log-recovery}, L18–L26). If the tags do not match, the follower replaces any inconsistent entries with the corresponding entries from the leader’s log.

The follower accepts the first valid leader reply that passes the epoch check. It appends the corresponding suffix and completes recovery. Any subsequent replies are ignored. If no valid leader reply is received, a standard leader election eventually produces a leader under partial synchrony. The follower can then complete the recovery process.

\subsubsection{Log Recovery for the Leader}
The leader recovery path is designed to minimize the system’s unavailability following a leader crash. If the crashed leader waits for a timeout and a subsequent election, the system remains unavailable during that interval. In \sysname, the crashed leader can resume service after recovery.

The leader first checks $TC_{role}$ to confirm that it crashed during a leader phase. It then advances the epoch and broadcasts a \msg{LeaderRecover} message (Alg.~\ref{algo:unified-log-recovery}, L13–L16). A node replies only if it can vote in that epoch. Before sending its reply, the node catches up to the new epoch and records that it has voted for the leader of that epoch. In this way, every follower that contributes to the recovery quorum will not vote for another candidate in the same epoch (Alg.~\ref{algo:unified-log-recovery}, L28–L38).

The recovering leader collects $f{+}1$ replies and selects the log with the most up-to-date last-entry tag. This quorum is sufficient to preserve all committed entries. By quorum intersection, any entry committed before the crash must appear in at least one of the replies. Selecting the log with the latest last-entry tag ensures that all committed entries are included, preventing any loss (Alg.~\ref{algo:unified-log-recovery}, L40–L48).

If the recovering leader resumes in the old epoch without advancing, entries that were not yet committed could cause inconsistencies. Depending on which followers respond, such entries may or may not be included in the $f{+}1$ selected log. If the leader resumes in the old epoch after losing these entries, it could propose conflicting entries at the same index under the same epoch-index tag. Moving to a new epoch eliminates this ambiguity. After synchronization, the leader extends the recovered log in the new epoch, and any unrecoverable, uncommitted suffix from the old epoch is safely discarded.

\subsubsection{Log Recovery Acceleration} 
\sysname further reduces recovery and normal-case overhead through two optimizations enabled by protocol-guided log recovery.

\bheading{Unblocking Log Recovery.}
If the leader remains unchanged, it tracks the next entry index for each follower. For follower $p_i$, this index is denoted as $idx_i$. After unsealing its log from persistent storage, $p_i$ can complete recovery once it has caught up to $idx_i$, without synchronizing with the leader’s latest log. This mechanism accelerates recovery by eliminating unnecessary synchronization.

\bheading{Background Log Persistence.}
\sysname moves log persistence off the commit critical path to reduce synchronous I/O during normal-case replication. Under high load, a node keeps newly appended log entries in memory and flushes them to disk in the background. This design is safe because log recovery relies on the cluster rather than the persistent storage as the sole source. By contrast, some systems require every committed log entry to be durably written to local storage. Such systems cannot use background persistence. Losing unflushed entries in these cases would compromise safety.

\section{Correctness Analysis and Verification}

\subsection{Correctness Analysis}\label{subsec:analysis}

TEE integrity and the underlying CFT consensus protocol together guarantee consensus safety and liveness during normal execution~\cite{castro1999pbft}; our goal is therefore to show that \sysname's recovery procedure preserves the three properties of \ssecref{subsec:problemstatement} in the presence of rollback attempts on sealed storage. We collect primitive assumptions inherited from the TEE and protocol-level invariants maintained by \ssecref{subsec:metadataDesign}--\ssecref{subsec:logDesign}, state three formal theorems capturing those properties, introduce two supporting lemmas, and prove the theorems in dependency order.

By the threat model in \ssecref{subsec:sysmodel}, trusted hardware provides the following primitive guarantees:
\begin{packeditemize}
\item \textbf{A1 (Counter monotonicity).} For any counter we use, the value returned by \ReadTC{} is monotonically non-decreasing across crashes, and $\IncTC{}$ returns a value strictly greater than any value previously returned for the same counter.
\item \textbf{A2 (Seal authenticity).} A successful \UnSeal{} returns a plaintext that was produced inside an enclave with the matching measurement at some earlier time.
\item \textbf{A3 (Enclave integrity).} Code running inside an enclave executes as specified; the adversary can delay or drop enclave messages, but cannot tamper with authenticated messages.
\end{packeditemize}
Each node maintains two trusted counters: $TC_{md}$ for metadata updates and $TC_{role}$ for role tracking. 

\bheading{Protocol Invariants.} The procedures in \ssecref{subsec:metadataDesign}--\ssecref{subsec:logDesign} maintain the following invariants at every correct node. 
\begin{packeditemize}
\item \textbf{I1 (Counter--epoch binding).} After every successful metadata update, the sealed metadata satisfies $md.\textit{epoch}=ctr$, where $ctr$ is the value returned by the corresponding $\IncTC(md)$; the in-memory epoch always equals $md.\textit{epoch}$.
\item \textbf{I2 (Mismatch repair).} During recovery, if $\ReadTC(md)$ differs from the sealed counter value, the node sets its epoch to $\ReadTC(md)$ and disables voting in that epoch.
\item \textbf{I3 (Fresh epoch for leader recovery).}
A recovering leader uses $TC_{role}$ to identify that it crashed in a leader phase.
Before resuming, it advances metadata to a new epoch and obtains at
least $f{+}1$ \msg{ReplyRecover} responses from nodes that were eligible to vote in the new epoch.
\item \textbf{I4 (Safe log selection in leader recovery).} Before a recovering leader proposes in a fresh epoch, it collects \msg{ReplyRecover} responses from at least $f{+}1$ nodes and adopts the responder's log with the most up-to-date last-entry tag $(ep, \textit{lastIdx})$.
\item \textbf{I5 (Log prefix consistency).}
A follower appends a suffix from a leader only if the suffix's predecessor tag $(prevEp, prevIdx)$ matches the follower's local entry at $prevIdx$. If the tags do not match, the follower truncates the conflicting suffix and retries from an earlier matched tag. This rule ensures that log synchronization preserves a common prefix between the leader and follower.
\item \textbf{I6 (Leader append-only).} While acting as a leader, a node only appends at indices strictly greater than its current $\textit{lastIdx}$ and never rewrites an earlier entry of its own log.
\item \textbf{I7 (Up-to-date vote rule).} A valid vote is generated only if the candidate's last-entry tag $(lastEp, lastIdx)$ is lexicographically greater than or equal to the voter's own.
\end{packeditemize}
Invariants I1--I2 follow from the metadata update/recovery procedure of \ssecref{subsec:metadataDesign}; I3--I4 follow from the leader-recovery steps of \ssecref{subsec:logDesign}; I5--I7 are inherited from the underlying leader-based CFT consensus (\eg, Raft~\cite{raft2014}).

\begin{theorem}[Election Safety]\label{thm:election}
In any epoch $ep$, at most one node can act as a leader.
\end{theorem}

\begin{theorem}[Leader Completeness]\label{thm:completeness}
If an entry is committed in epoch $ep$, every leader that assumes leadership in any later epoch $ep' > ep$ contains that entry in its log.
\end{theorem}

\begin{theorem}[Recovery Liveness]\label{thm:liveness}
Every recovering, uncorrupted node eventually completes its recovery procedure.
\end{theorem}

We now state two supporting lemmas and then prove Theorems~\ref{thm:election}--\ref{thm:liveness} in dependency order.

\begin{lemma}[Recovered Epoch Monotonicity]\label{lem:epochmono}
For any correct node, the epoch after any crash recovery is no smaller than its epoch immediately before the crash.
\end{lemma}

\begin{proof}
Let $p_i$ be a correct node whose pre-crash epoch is $ep$. By I1, $ep$ equals the value $ctr$ returned by $p_i$'s last successful $\IncTC(md)$ (with $ep = 1$ if $p_i$ has never completed a metadata update). By A1, $\ReadTC(md)$ at recovery returns some $v \geq ctr = ep$. Recovery then proceeds as follows: in the \emph{match} case, the restored epoch equals the sealed $ctr' = v$ (which by A2 and I1 is authentic) and thus equals $v \geq ep$; otherwise (\emph{mismatch} or \UnSeal{} failure), I2 sets the epoch to $v \geq ep$.
\end{proof}

\begin{lemma}[Vote Uniqueness per Epoch]\label{lem:voteunique}
In each epoch, a node can grant a valid vote to at most one node.
\end{lemma}

\begin{proof}
In \sysname, a valid vote is either a normal election vote or a recovery vote carried by a \msg{ReplyRecover} message. 
Fix an epoch $ep$ and a node $p_i$. 
In a normal election, $p_i$ can vote for $ep$ only when it first advances its epoch to $ep$. I1 and A1 ensure that $p_i$ has at most one voting opportunity in $ep$.
Upon recovery, $p_i$ either restores the sealed vote or disables voting for the recovered epoch by I2. 
Thus, $p_i$ cannot cast another vote in $ep$. 
In leader recovery, sending a \msg{ReplyRecover} counts as a recovery vote. By I3, $p_i$ sends this message only if it is still eligible to vote in $ep$, and it then becomes non-voting. Therefore, $p_i$ can grant a valid vote to at most one node in epoch $ep$.
\end{proof}

\begin{proof}[Proof of Theorem~\ref{thm:election}]
Assume, for contradiction, that two distinct nodes $p_a \neq p_b$ act as leaders in the same epoch $ep$. Let $Q_a$ and $Q_b$ be their vote quorums, where a recovery quorum of \msg{ReplyRecover} messages is treated as a vote quorum. Then $|Q_a| \geq f+1$ and $|Q_b| \geq f+1$. Since $n=2f+1$, we have $Q_a \cap Q_b \neq \emptyset$. Let $p_c \in Q_a \cap Q_b$. Then $p_c$ grants a valid vote to both $p_a$ and $p_b$ in epoch $ep$, contradicting Lemma~\ref{lem:voteunique}.
\end{proof}

The following technical lemma underlies the proof of Leader Completeness.

\begin{lemma}[Leader Append-Only Preservation]\label{lem:leaderholds}
If $L_T$ becomes the leader with entry $en$ at index $k$ in its log, then $L_T$ preserves $en$ at index $k$ throughout its leadership. Moreover, any new entry proposed by $L_T$ is appended after index $k$.
\end{lemma}

\begin{proof}
When $L_T$ becomes the leader, its last log index is at least $k$. By I6, $L_T$ only appends new entries after its current last log index and never rewrites earlier entries while it remains the leader. Therefore, $en$ remains at index $k$, and every new entry proposed by $L_T$ is appended after index $k$.
\end{proof}

\begin{proof}[Proof of Theorem~\ref{thm:completeness}]
Let $en$ be committed at index $k$ in epoch $ep$ by leader $L$. By the commit rule, $en$ is replicated to a quorum $Q_c$ with $|Q_c|\ge f+1$, each storing $en$ at $(ep,k)$.

We prove by strong induction on $T > ep$ that every leader $L_T$ in epoch $T$ holds $en$ at index $k$ throughout its tenure.

\emph{Base case ($T=ep+1$).} Consider the first leader $L_{ep+1}$ after $en$ is committed. By the quorum intersection argument, $L_{ep+1}$ adopts a log containing the committed prefix including $en$ at $k$. Lemma~\ref{lem:leaderholds} then ensures $en$ remains at $k$.

\emph{Induction hypothesis.} Assume that every leader $L_{T'}$ for $ep < T' \le T$ contains $en$ at $k$ throughout its tenure.

\emph{Induction step ($T \to T+1$).} Consider leader $L_{T+1}$.  

- If $L_{T+1}$ assumes leadership via a normal election, let $Q_v$ be its vote quorum. Since $|Q_v|,|Q_c|\ge f+1$, they intersect at some $r$. By induction hypothesis and I5, $r$'s log contains $en$ at $k$. By I7, $r$ votes only if $L_{T+1}$'s last-entry tag $\succeq r$'s, so I5 ensures $L_{T+1}$'s log contains $en$.

- If $L_{T+1}$ assumes leadership via recovery, let $Q_r$ be the recovery-quorum with $Q_r\cap Q_c\neq \emptyset$ and $r$ in it; by I4, $L_{T+1}$ adopts the log with largest last-entry tag $\ge r$'s log, which I5 ensures preserves the committed prefix with $en$ at $k$.

Thus, by induction, every leader in epoch $T+1$ contains $en$ at index $k$ before proposing, completing the proof.
\end{proof}

\begin{proof}[Proof of Theorem~\ref{thm:liveness}]
Consider a correct node $p_i$ beginning recovery.

\emph{(1) Metadata recovery.} $p_i$ invokes \UnSeal{} and $\ReadTC(md)$, both $O(1)$ TEE operations. On match, it restores $md'$ (I1); on mismatch or \UnSeal{} failure, it applies I2. In either case, this phase terminates.

\emph{(2) Log recovery.} If $p_i$ is a follower, after \textsf{GST} a current or newly elected leader is reachable within bounded time. The leader replies to \msg{FollowerRecover} with the missing suffix, which $p_i$ appends to complete recovery. If $p_i$ is a recovering leader, it broadcasts \msg{LeaderRecover}; after \textsf{GST}, either $f+1$ nodes in the fresh epoch respond with \msg{ReplyRecover}, in which case I4 lets $p_i$ adopt the most up-to-date log and resume proposing, or another leader is eventually established and $p_i$ recovers via the follower path. A nonce in recovery messages prevents responses from being replayed.

Each phase completes in bounded time after \textsf{GST}, so $p_i$ eventually rejoins the system.
\end{proof}

\subsection{Formal Verification.}
We develop a Maude~\cite{Maude} specification of \sysname atop Braft and model check it against both safety and liveness guarantees expressed in linear temporal logic (see Appendix~\ref{appen:maude} for details). We chose Maude as it is a well-established formal specification language and analysis framework that has been successfully applied to a broad range of distributed and networked systems~\cite{bobba2018survivability,10.1145/3603269.3604870,DBLP:series/isc/Chuat22,DBLP:journals/pacmpl/LiuMOZB22}. Within the explored bound, the model checker reports no counterexample to the three recovery properties and protocol liveness.

\section{Implementation} \label{sec:implementation}
We implement \sysname on top of Braft~\cite{braft} and ZooKeeper~\cite{ZookeeperApache}, which use the leader-based CFT consensus protocols Raft and Zab, respectively.\footnote{
Source code is at \url{https://github.com/Artifacts2026/CHIMERA}.} We refer to these implementations as \sysname-B and \sysname-Z. For TEE support, we port \sysname to VM–based TEEs, \ie, Intel TDX, and apply additional optimizations to mitigate the performance overhead of TEEs. (The rationale for choosing Braft, ZooKeeper, and TDX is discussed in~\ssecref{sec:intro}.)

We adopt the software-based counter TIKS~\cite{engraft} for metadata recovery. Moreover, we use Narrator-Pro~\cite{narrator-pro} as the
trusted counter to safely update the counter within a single round of communication.
For sealing functionality, we invoke Intel SDK~\cite{intelSDKgithub} to retrieve the measurement ($mr\_td$) of TDX, and derive a stable sealing key from it.
This key is used to encrypt and decrypt data for the persistent storage outside TEEs. 

\bheading{Implementation atop Braft.} We list required modifications on Braft to realize \sysname:  
\begin{packeditemize}
    \item \textbf{Metadata.} Raft’s metadata consists of two fields: \currentTerm and \votedFor (\ssecref{subsec:character}). In Braft, these fields are combined into a single on-disk structure called \raftMeta. In \sysname, we attach a counter value to each update of \raftMeta to support precise recovery.
 
    \item \textbf{Log.} In \sysname-B, new log entries are first buffered in memory. Their persistence is handled by a background thread that runs during idle periods. This design reduces critical-path latency and allows higher throughput compared to Braft (\ssecref{sec:normalEvaluation}). For recovery, we reuse Braft’s existing \textit{replicator} component to transmit log entries to recovering nodes.
\end{packeditemize}

\bheading{Implementation atop ZooKeeper.} 
Modifications on ZooKeeper, including the log, are similar to those on Braft. 
Thus, we only list the differences: 

\begin{packeditemize}
    \item \bheading{Metadata.} Unlike Raft’s single-phase election, the Zab protocol uses a two-phase process: discovery and synchronization. In this process, a prospective leader is first nominated and is only promoted after receiving acknowledgments from a majority of followers~\cite{zab}. Specifically, after Fast Leader Election (FLE), the newly elected leader enters the discovery phase. It proposes a fresh epoch $ep^*$, which is greater than any previously observed. Each follower persists this value as \acceptedEpoch and returns an acknowledgment to the leader.

    In Zab, followers do not need to record which node they voted for. Election safety is ensured instead by two metadata fields: \acceptedEpoch and \currentEpoch. Together, these fields guarantee that at most one leader is recognized by the majority. Because the two fields are updated under different conditions, we use separate counters to protect them independently.
\end{packeditemize}

\begin{figure*}[t]
    \centering
    \subfloat[Normal case, WAN]{%
        \includegraphics[width=0.24\linewidth]{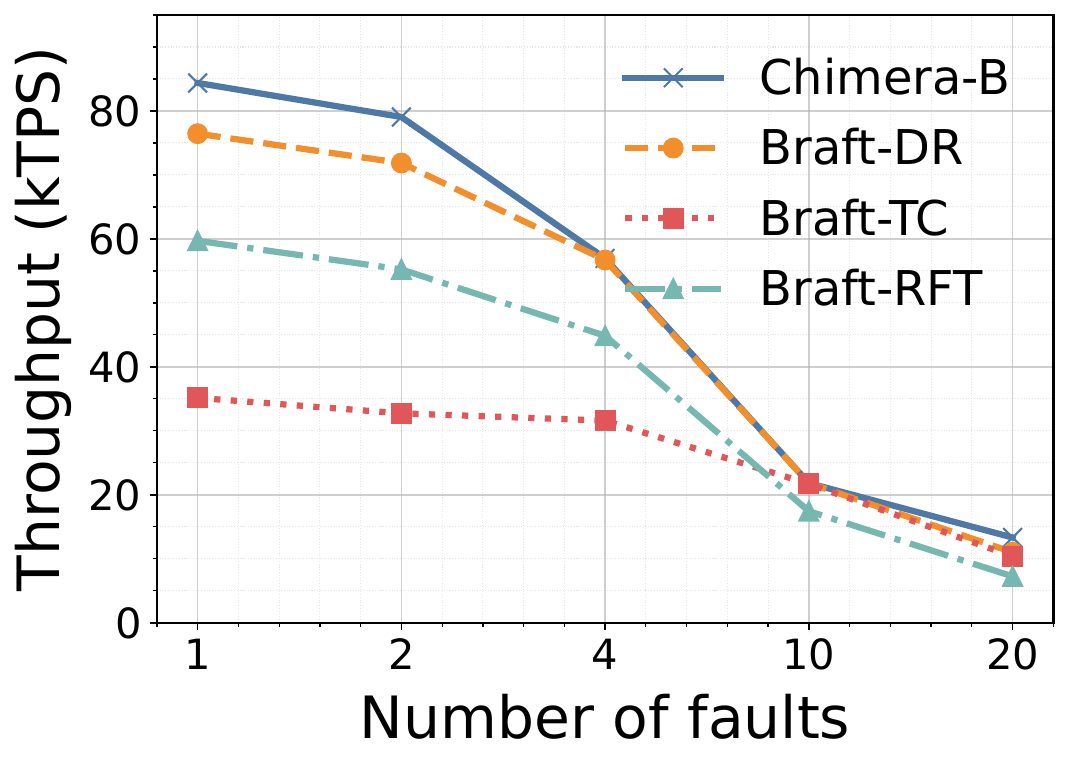}%
        \label{fig:t-wan-node}%
    }\hfill
    \subfloat[Normal case, WAN]{%
        \includegraphics[width=0.24\linewidth]{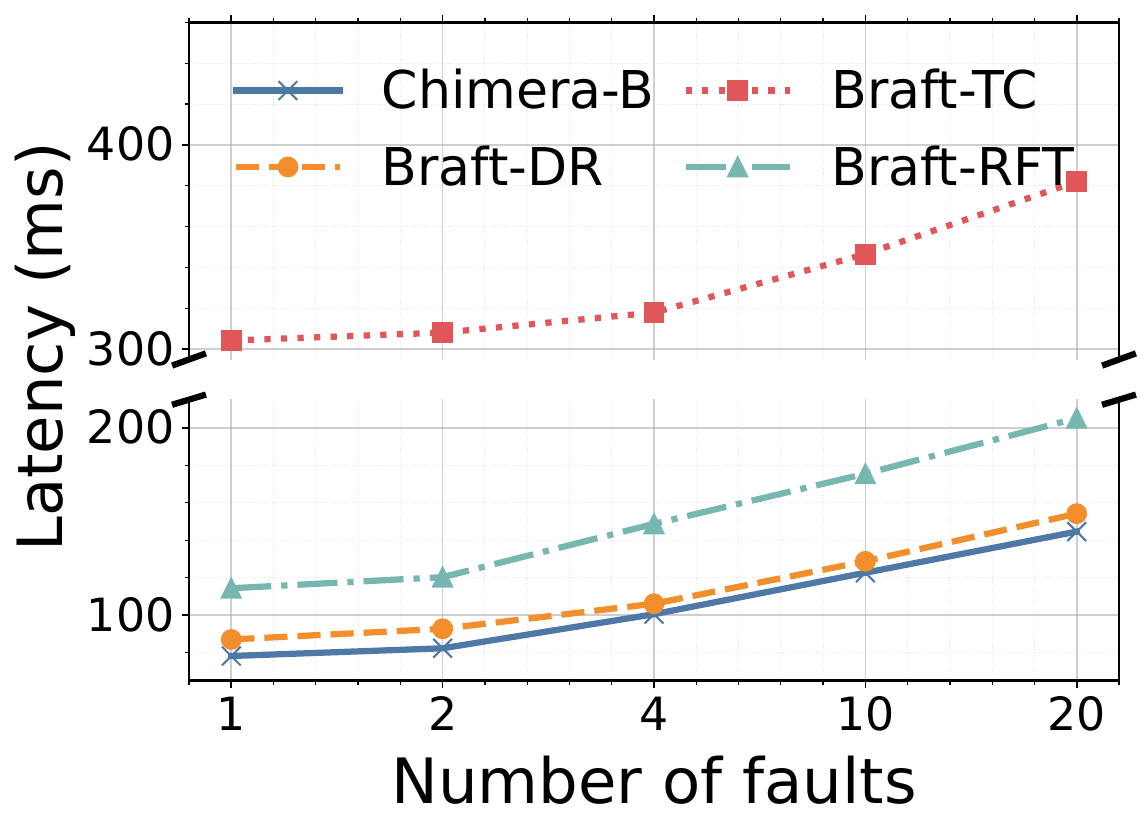}%
        \label{fig:l-wan-node}%
    }\hfill
    \subfloat[Normal case, LAN]{%
        \includegraphics[width=0.24\linewidth]{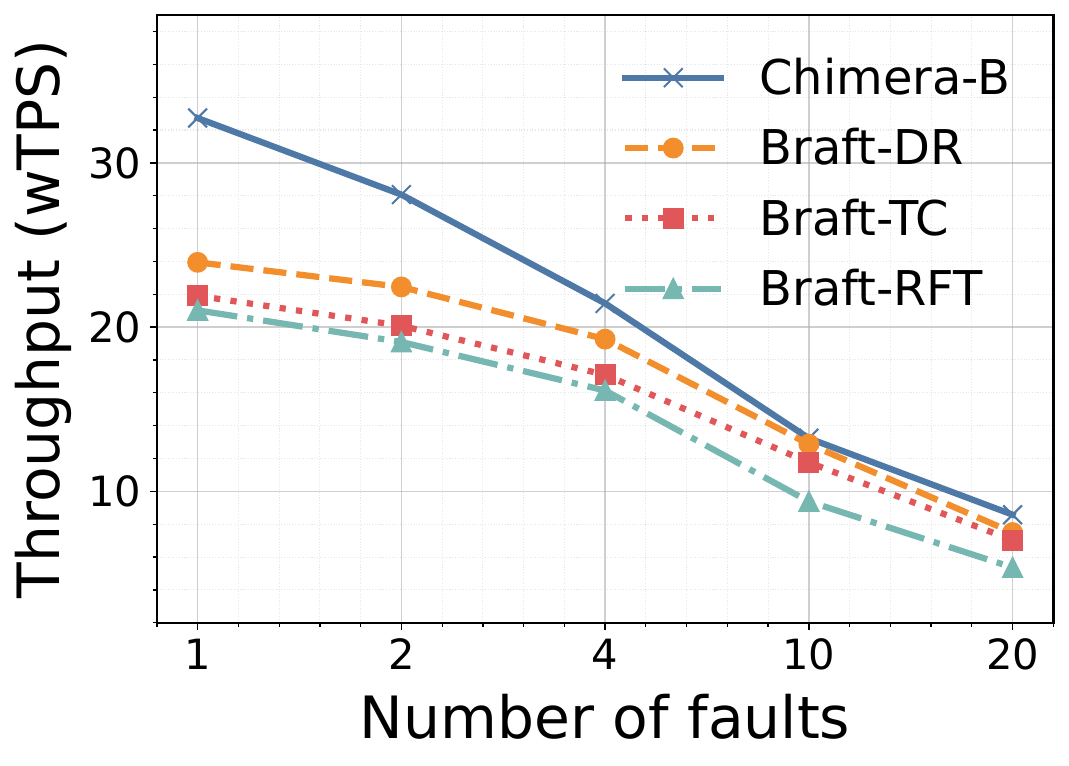}%
        \label{fig:t-lan-node}%
    }\hfill
    \subfloat[Normal case, LAN]{%
        \includegraphics[width=0.24\linewidth]{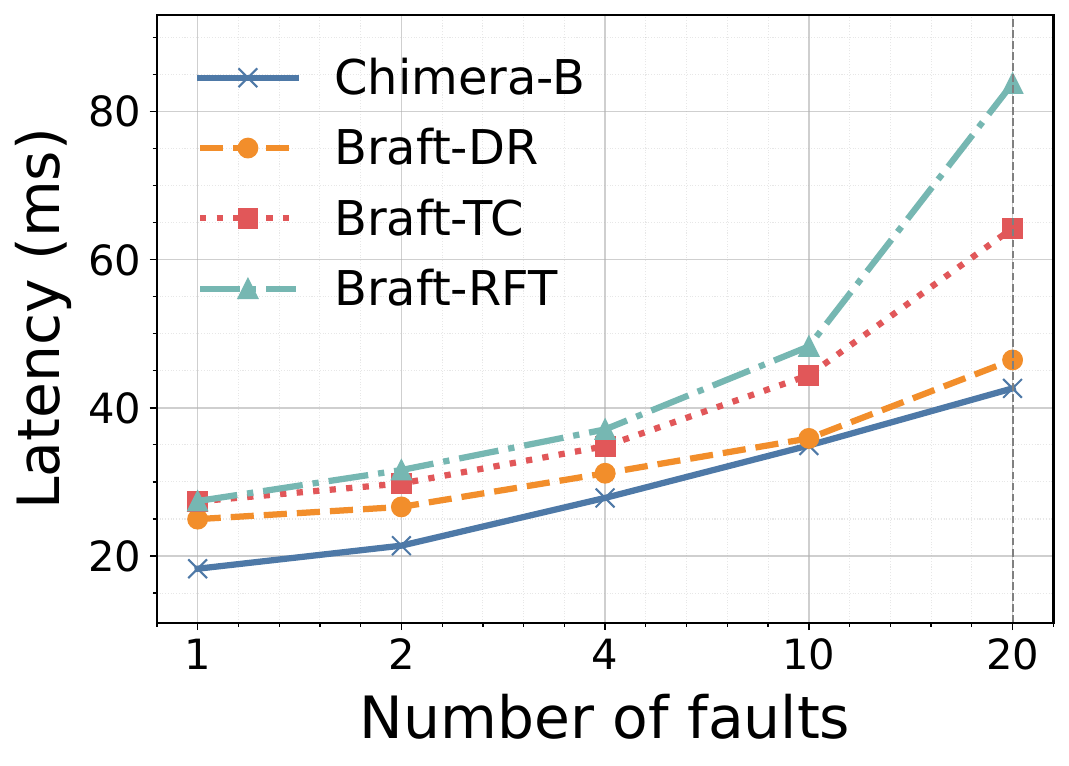}%
        \label{fig:l-lan-node}%
    }
    \caption{Throughput and latency comparisons with varying nodes in WAN and LAN.}
    \label{fig:wan}
    \vspace{-4mm}
\end{figure*}

\section{Evaluation} \label{sec:exp}
We evaluate the performance of \sysname-B (built on Braft~\cite{braft}) and \sysname-Z (built on ZooKeeper~\cite{ZookeeperApache}). Since ZooKeeper comprises additional components beyond consensus and its performance is influenced by multiple factors, we focus our main evaluation on \sysname-B, especially in comparison with its counterparts, while providing an overhead profiling of \sysname-Z.  
We consider four recovery taxonomies (\ssecref{sec:taxonomy}) and adopt them on Braft as baselines. For consistency, all protocols running inside TDX perform disk I/O through sealing.

\begin{packeditemize}
    \item \textbf{Braft-TC} 
    adopts a software-based counter to protect state updates as Engraft~\cite{engraft}. Each counter increment requires two rounds of broadcast to advance the state.

    \item \textbf{Braft-RFT} 
    follows FlexiBFT~\cite{gai2021dissecting} by increasing the node size from $2f+1$ to $3f+1$ and employs a TC only at the leader to prevent equivocating proposals.
 
    \item \textbf{Braft-RC} 
    adopts the original reconfiguration design of Raft, implemented through joint consensus~\cite{raft2014}.
    
    \item \textbf{Braft-DCR} 
    follows Achilles~\cite{achilles} to ask a recovering node to skip two epochs for safety. In Braft, the Pre-Vote mechanism bounds the growth of epoch, making DCR feasible. 
\end{packeditemize}
In addition to the four baselines, we introduce Braft-Direct Recovery (Braft-DR) to evaluate the overhead of rollback resilient solutions. In Braft-DR, a node restores its state directly from sealed data on untrusted storage.

We evaluate \sysname-B against all baselines under both fault-free and faulty scenarios to answer three questions:

\begin{packeditemize}
    \item[\textbf{Q1:}] How does \sysname perform with varying nodes in WAN and LAN compared to its counterparts?  (\ssecref{sec:normalEvaluation})
    
    \item[\textbf{Q2:}] What is the performance of \sysname's recovery protocol, where do the primary bottlenecks lie, and how does it compare to prior approaches? (\ssecref{sec:recoveryEvaluation})
    
    \item[\textbf{Q3:}]
    How much overhead do TEE-related operations introduce,  
  and how effective are our optimizations? (\ssecref{sec:teeEvaluation})
\end{packeditemize}

\subsection{Experimental Setup}
We conducted all experiments on a public cloud platform using up to 61 Intel TDX-enabled instances, with one instance per node. Each node ran on a dedicated virtual machine provisioned with 4 vCPUs and 16 GB of RAM, running Linux kernel~5.10 LTS (64-bit).

We evaluate \sysname under two deployment scenarios: Local Area Network (LAN) and Wide Area Network (WAN). Both are configured using the Linux \texttt{tc} tool for traffic shaping. In the LAN setting, the per-node bandwidth is limited to 10 Gbps, and the inter-node RTT stays below 1 ms. In the WAN setting, we emulate wide-area conditions by limiting bandwidth to 1 Gbps per node and enforcing a 40 ms round-trip latency with $\pm$4 ms jitter. Note that WAN evaluation is performed in emulation because TEE-enabled instances are restricted to specific cloud regions; this reflects deployment constraints rather than a design flaw.

\bheading{Parameters and Metrics.} 
We vary the fault-tolerance parameter $f \in \{1, 2, 4, 10, 20\}$. The leader processes client requests in batches of 256, with each log entry containing a 256~B payload. To balance memory consumption and network utilization, the system bounds the number of in-flight log entries at 5120. 
In addition to \textit{throughput} and \textit{latency}, we measure  \textit{recovery time}, defined as the interval from when a faulty node restarts until it fully recovers and resumes service.

\subsection{Fault-Free Performance} \label{sec:normalEvaluation}
This section evaluates \sysname-B's performance of log commitment in the fault-free scenarios with no failures or leader changes. We compare \sysname-B with Braft-TC, Braft-RFT, and Braft-DR. As Braft-DCR and Braft-RC introduce no additional overhead in normal-case operations compared to Braft-DR, their results are not presented in detail. Due to space constraints, we present the throughput–latency scalability of \sysname-B in Appendix~\ref{appen:addExp}.

\subsubsection{Performance in WAN}
We evaluate throughput and latency under a WAN deployment while varying the fault threshold $f$ (Fig.~\ref{fig:t-wan-node} and \ref{fig:l-wan-node}). In this setting, communication overhead and bandwidth limits dominate performance.

During each transaction, Braft-TC requires two additional communication rounds to interact with its TC for securely recording log changes. According to Raft's replication semantics, this results in a total of five communication rounds to complete a transaction, incurring substantial latency. Braft-RFT, by contrast, adopts a $3f+1$ configuration, which, for the same fault-tolerance level $f$, entails more nodes and thus higher processing overhead. Moreover, an additional counter update during leader persistence incurs one more communication round, which further degrades WAN performance.

Both \sysname-B and Braft-DR operate with $2f+1$ nodes and incur no additional normal-case overhead, which explains their relatively higher throughput in this environment. Nevertheless, 
\sysname-B's recovery mechanism does not rely on the completeness of on-disk logs, enabling it to safely defer disk writes to background operations and thus reduce normal-case commit latency. This design yields an average throughput improvement of approximately 10\% over Braft-DR; however, as communication dominates in the WAN scenario, the advantage from reduced I/O latency is less pronounced.

\subsubsection{Performance in LAN} 
We also evaluate the throughput and latency of \sysname in a LAN deployment to minimize the effect of network communication (Fig.~\ref{fig:t-lan-node} and~\ref{fig:l-lan-node}).
As the network communication cost is negligible in a LAN environment, local processing and I/O overhead become the dominant factors affecting performance. 

\sysname-B exhibits a pronounced throughput advantage over all other variants, particularly at low fault tolerance levels, due to its normal-case optimization that defers durable disk writes to background operations (as explained above). 
When $f = 1$, this optimization leads to a 68\% improvement over Braft-DR that follows the original Braft I/O semantics of synchronous persistence. 

Braft-TC incurs further penalties from its additional communication and trusted counter operations, though in the LAN setting, these penalties are partially masked by the low RTT.
Braft-RFT, with its $3f+1$ configuration, experiences a more significant throughput drop as fault increases, due to the larger quorum size and corresponding processing and message handling overhead. 
Overall, \sysname-B shows superior performance when network delays are not the primary bottleneck and storage-layer optimizations directly translate into substantial end-to-end performance gains.

\begin{table}[t]
    \centering
    \caption{Recovery overhead. 
    }
    \label{tab:log-overhead}
    \scalebox{0.9}{
    \begin{tabular}{@{}lccccc@{}}
        \toprule[1pt]
        Cost (s) & Braft-DR & Braft-TC & Braft-RC & Braft-DCR & \textbf{\sysname-B} \\
        \midrule
        Prep.  & N/A   &  0.01 &   0.01   & N/A     & \textbf{0.01} \\
        Sync. & 6.15    & 6.31    & 117.01 & 24.53 & \textbf{6.42 $(^{*} 18)$} \\
        Quie. & N/A    & N/A    & N/A & 2 epochs & \textbf{N/A} \\
        \hline
        Total    & 6.15  & 6.32  & 117.02 & 24.53 + 2 epochs & \textbf{6.43 $(^{*} 18.01)$} \\
        \bottomrule[1pt]
    \end{tabular}}
    \vspace{-4mm}
\end{table}

\subsection{Performance under Faults} \label{sec:recoveryEvaluation}
We evaluate \sysname-B against Braft-DR, Braft-TC, Braft-RC, and Braft-DCR to measure rollback-resilient recovery overhead under faults. Braft-RFT is omitted, as its recovery is identical to Braft-DR. We deploy 21 nodes in a LAN deployment and simulate failures by shutting down and restarting 10 nodes while continuously issuing client requests. Each node is equipped with 210 MB/s sequential read/write throughput and preloaded with 5 GB of log entries to emulate a large-scale deployment.

\bheading{Single-Node Recovery Latency.} Table~\ref{tab:log-overhead} presents the time for a recovering node to rejoin the protocol. To enable fair comparison, we divide the recovery process into three stages: preparation, synchronization, and quiescence, as below.

\begin{packeditemize}
    \item \textit{Preparation.} This stage includes the steps required to initialize state synchronization. In Braft-TC and \sysname-B, the node reads the trusted counter, which takes about 10 ms. In Braft-RC, the recovering node performs reconfiguration through two consensus rounds, also taking around 10 ms. The latency of RC depends on inter-node message delays.

    \item \textit{Synchronization.} This stage involves loading metadata and log entries, with the latter dominating the cost. In Braft-RC, the node must reconstruct its state from scratch, which takes about 110–120 seconds for 5 GB of data. For data-intensive applications such as blockchains, where the state can reach several terabytes (\eg, Bitcoin~\cite{nakamoto2012bitcoin}), the recovery time under RC can extend to hours or even days. 
    
    Braft-DCR's recovery requires disk loading ($\approx 6$s) and network synchronization ($\approx 18$s). By contrast, \sysname leverages its optimization to safely recover without network catch-up, which takes about 6.42s. Without the optimization, this stage takes about 18s. 

    \item \textit{Quiescence.} This stage ensures safety during rejoining. Braft-DCR requires a node to skip two epochs, which in practical deployments can range from tens of seconds to hours or even days.
\end{packeditemize}

\begin{figure}[t]
    \centering
    \includegraphics[width=0.9\linewidth]{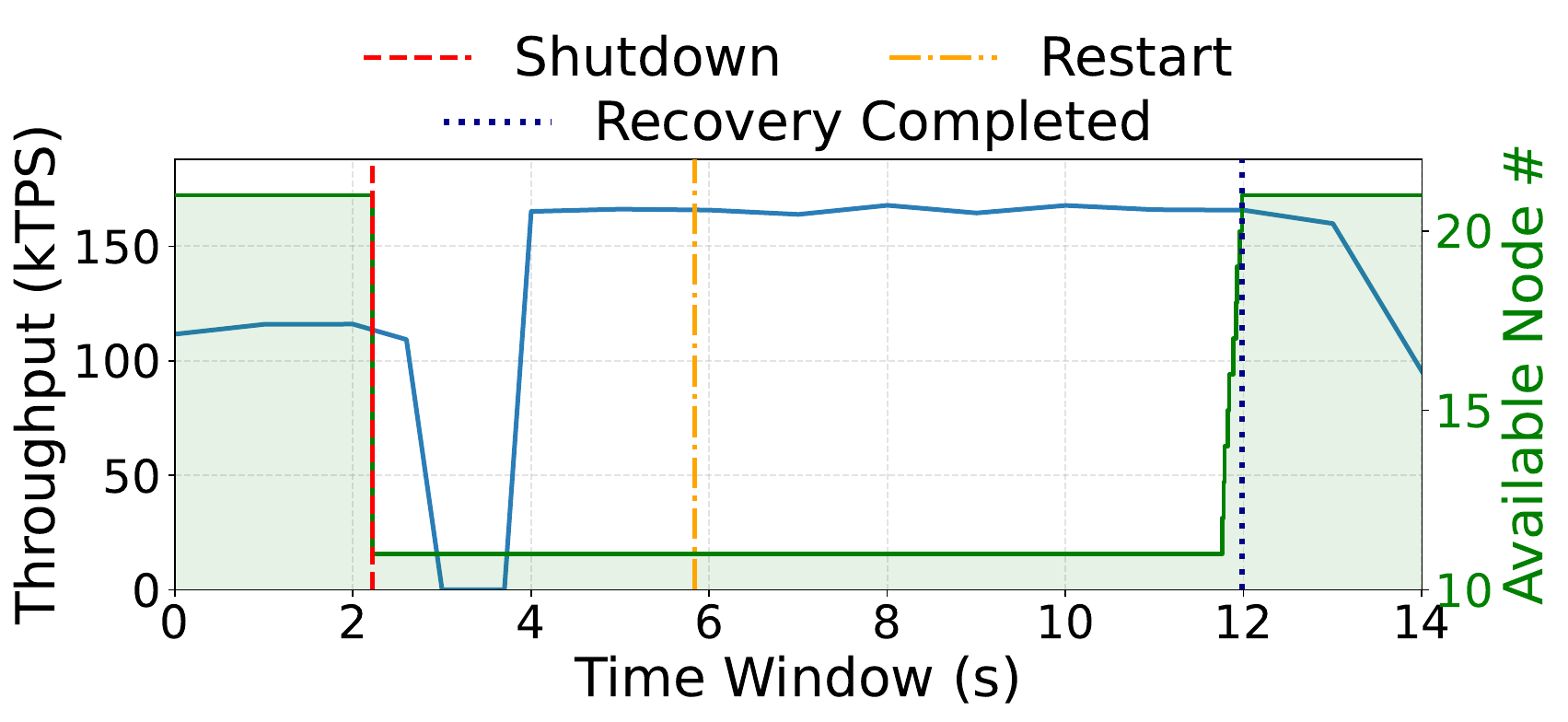} 
    \caption{The fault recovery process of \sysname-B.}
    \label{fig:fault}
    \vspace{-4mm}
\end{figure}

\bheading{Throughput under Recovering Faults.} We evaluate the system-level impact of recovery in terms of system throughput and available nodes, focusing on \sysname-B. Since Braft-DR, Braft-DCR (without quiescence stage), and Braft-TC exhibit recovery latencies comparable to \sysname-B, their effects are effectively captured by \sysname-B and are omitted here. For completeness, we defer the results for Braft-RC to Appendix~\ref{appen:addExp} due to space constraints.

Fig.~\ref{fig:fault} shows the throughput variation and the number of available nodes (\ie, who participate in consensus) of \sysname-B during faulty nodes' recovery. 
At the shutdown point (\ie, 2.2 seconds), the throughput briefly drops to zero as the leader handles connection failures. Once stabilized, the throughput surpasses the pre-failure baseline. This occurs because, although the quorum size is unchanged, the leader no longer replicates log entries to the failed nodes. Thus, the leader's available bandwidth is redistributed to the remaining nodes, reducing contention and increasing throughput. 

During recovery, faulty nodes reload their logs from disk, without affecting ongoing throughput. However, once recovery completes, these nodes lag behind because the leader continues to serve client requests. To rejoin replication, they must first catch up, consuming bandwidth and temporarily reducing throughput until synchronization finishes.

\subsection{Overhead Profiling} \label{sec:teeEvaluation}
To better understand the overhead of TEE-related operations, we compare variants of \sysname-B and \sysname-Z.

\begin{packeditemize}
    \item \textit{NoTEE.} It runs outside Intel TDX, serving as a baseline to measure TEE-related overhead.

    \item \textit{SyncWrite.} It uses synchronous disk writes instead of asynchronous persistence, isolating the cost of log persistence.
    
    \item \textit{NoEncrypt.} It disables memory encryption atop SyncWrite, revealing the overhead of cryptographic operations.
\end{packeditemize}

\bheading{\sysname-B.} \figref{fig:suba-breakdown} shows the throughput of \sysname-B and its variants as the number of faults ($f$) increases in a LAN setting, with all other parameters identical to the fault-free scenarios. Compared to \textit{NoTEE} variant, \sysname-B shows an 8–10\% slowdown, capturing the inherent cost of TEE execution. Relative to \textit{SyncWrite} variant, persistence alone adds roughly 15–20\% overhead, independent of encryption. Finally, the difference between the \textit{SyncWrite} and \textit{NoEncrypt} variants quantifies the cost of cryptographic sealing, resulting in an additional 10–15\% performance degradation.

\bheading{\sysname-Z.} \figref{fig:subb-breakdown} presents the maximum throughput of \sysname-Z and its variants. While \sysname-Z exhibits a similar trend, the performance gaps are smaller. This is because the additional components and coordination overhead in ZooKeeper limit peak throughput,  masking much of the relative impact of TEE execution and persistence operations.

\begin{figure}[t]
  \centering
  \subfloat[\sysname-B breakdown]{%
      \includegraphics[width=0.48\columnwidth]{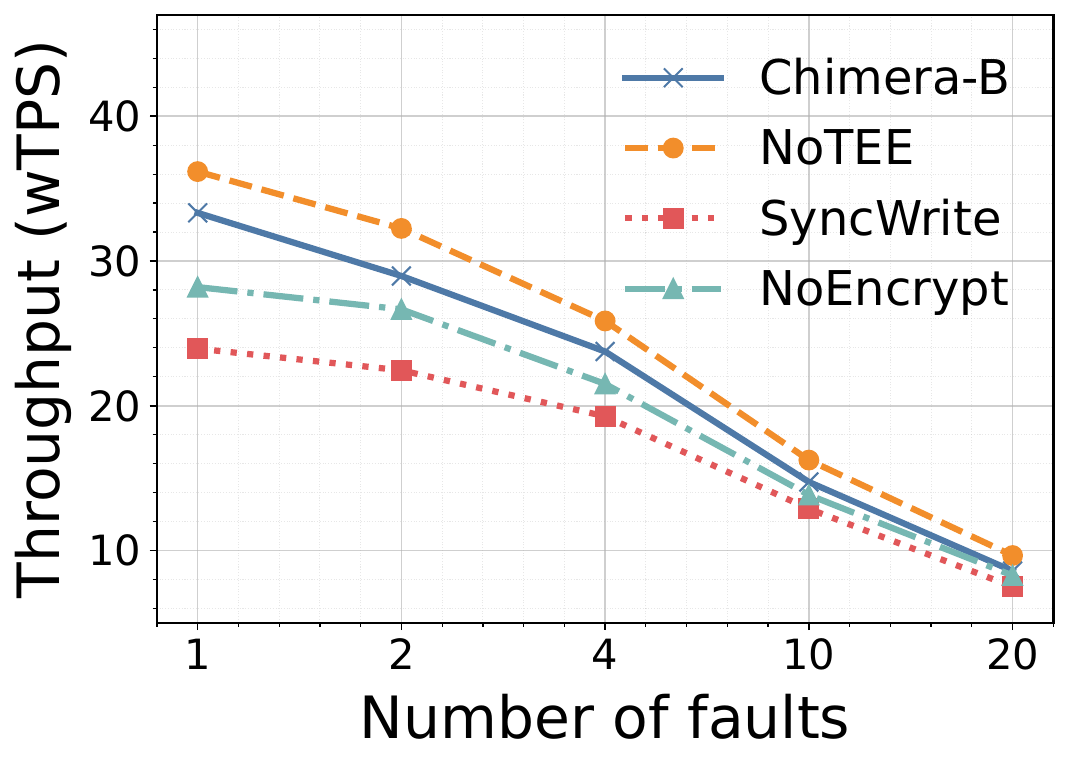}%
      \label{fig:suba-breakdown}%
  }\hfill
  \subfloat[\sysname-Z breakdown]{%
      \includegraphics[width=0.48\columnwidth]{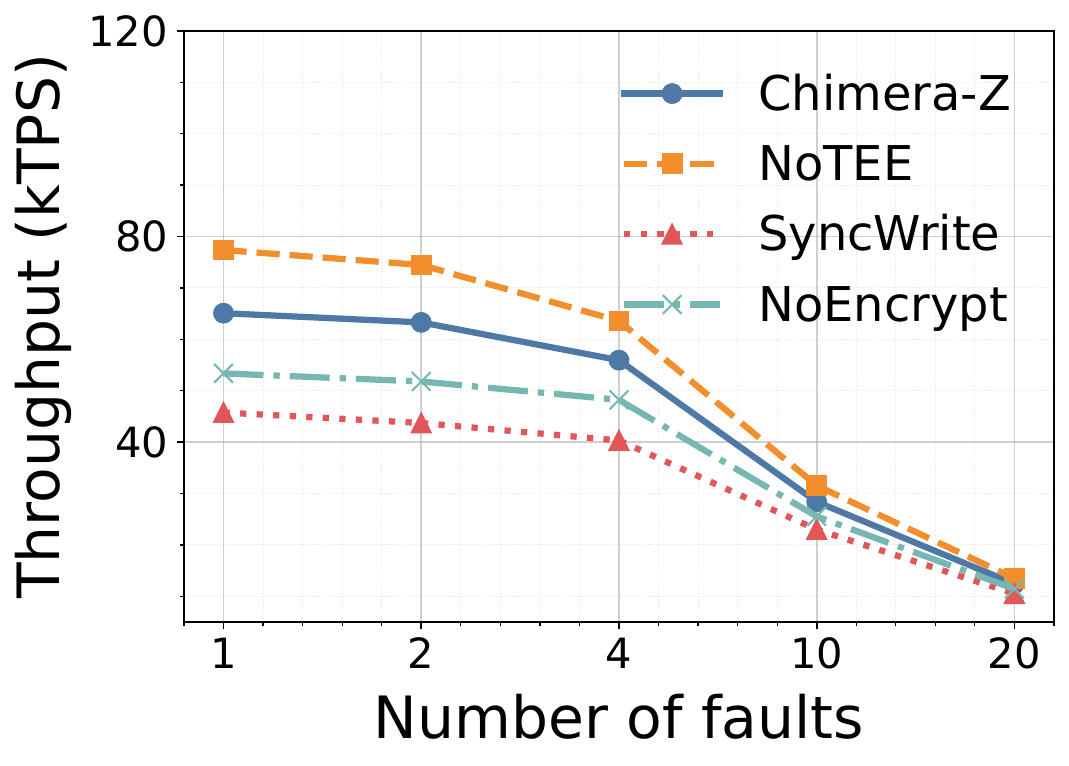}%
      \label{fig:subb-breakdown}%
  }
  \caption{Overhead profiling of TEE-related execution.}
  \label{fig:breakdown}
  \vspace{-4mm}
\end{figure}

\section{Related Work} \label{sec:RelatedWork}
We discuss prior work on confidential computing, TEE-assisted BFT consensus, and trusted counters. 

\bheading{Confidential BFT Service.} 
Confidential BFT services have recently attracted significant attention from industry and academia, driven by the explosive growth of cloud and decentralized applications. Notable industrial examples include SVR3~\cite{svr3} that ports Raft into TEEs to secure private key management, and Azure that provides Confidential Ledger service~\cite{Microsoft} atop CCF. 
Meanwhile, in academia, SecureKeeper~\cite{securekeeper} is among the first to use TEEs to protect metadata confidentiality in cloud settings with minimal changes to the Zab protocol~\cite{zab}. 
Similarly, Brandenburger~\etal~\cite{HyperSrds} integrate Intel SGX into Hyperledger Fabric~\cite{hyperledger} to secure smart contract execution.
CCF~\cite{2023ccf} uses enclaves to maintain a distributed key-value store and runs Raft~\cite{raft2014} to achieve low latency and tolerate a minority of Byzantine faults. 

However, most systems do not consider TEEs' rollback attacks. Engraft~\cite{engraft} first identified this threat and introduced TIKS, \ie, software-based counters, to enforce trusted counters for rollback protection. 
This approach corresponds to TC (as introduced in \ssecref{sec:taxonomy}), which represents the most general solution. More details of the trusted counter are introduced shortly.
Later versions of CCF~\cite{ccf2025} address the issue by reconfiguration, referred to as RC. However, these solutions degrade either performance or availability (\ssecref{sec:taxonomy}).

\bheading{TEE-Assisted BFT Consensus.} 
Unlike confidential BFT consensus that ports whole consensus protocols into TEEs, TEE-Assisted BFT consensus~\cite{Behl:2017:HSS, Liu:2019:SBC, zhang2021efficient, decouchant, gai2021dissecting, decouchant2024oneshot, achilles} usually utilizes TEEs to provide some trusted functions, such as the append-only log and monotonic counter, to minimize Trusted Computing Base (TCB). 
These trusted functions can prevent Byzantine nodes from equivocating messages, resulting in better scalability in terms of smaller quorum size and shorter transaction latency. 
Recently, FlexiBFT~\cite{gai2021dissecting} and Achilles~\cite{achilles} identified TEEs' rollback issues in TEE-Assisted BFT consensus and proposed RFT and NVR as solutions, respectively. However, RFT relaxes the tolerance, while NVR weakens the system's tolerance. See more details in \ssecref{sec:taxonomy}.

\section{Conclusion and Future Work} \label{sec:conclusion}
We systematically analyze existing TEE rollback-resilient solutions, establishing a taxonomy to assess their suitability for confidential BFT consensus. 
Building on the insights, we propose \sysname, a hybrid recovery framework that tailors recovery strategies for persistent state. 
We prove \sysname's correctness and complement our proofs with formal verification of its Braft design. 
We implement proof-of-concept prototypes atop Raft and ZooKeeper using Intel TDX, and our extensive evaluation demonstrates that \sysname delivers superior performance.

Next, we discuss our approach's limitations and potential extensions.
First, \sysname focuses on rollback-resilient recovery and does not currently support dynamic reconfiguration. Integrating reconfiguration with recovery is challenging because configuration updates are themselves stored in the replicated log. During recovery, a node must know the current configuration to safely recover the log, while the latest configuration may only exist inside the log being recovered. We leave the integration of reconfiguration into \sysname as future work. Second, although we focus on confidential BFT consensus, the core insight of \sysname, \ie, using protocol-level semantics to design tailored TEE recovery, extends to other confidential computing systems, including confidential MapReduce frameworks~\cite{schuster2015vc3}, federated learning~\cite{quoc2021secfl}, and encrypted databases~\cite{priebe2018enclavedb}. More broadly, separating critical metadata from bulk state and customizing recovery accordingly may benefit a wider range of stateful TEE applications.

\clearpage

\section*{Ethics Considerations}
This work studies rollback-resilient recovery for confidential BFT consensus systems. 
Our experiments are conducted on controlled cloud testbeds using synthetic workloads and do not involve human subjects, personal data, or attacks on third-party systems. 
The evaluated vulnerabilities are analyzed under an abstract threat model, and the artifacts are intended solely for research and reproducibility.
\bibliographystyle{IEEEtran}
\bibliography{bib}

\appendix

\subsection{Supplementary Background}
\subsubsection{Trusted Counter and Extensions} \label{appen:TC}
Trusted counter is the most representative and common method for rollback prevention. Generally, there are two types of counters: hardware-based and software-based. The first includes SGX counter~\cite{sgx_counter}, TPM counter~\cite{ADAM}, and TPM NVRAM~\cite{Ariadne, ICE, Memoir}. 
Hardware-based counters usually have poor performance, \ie, long latency (\eg, tens of milliseconds) for read and write operations and limited write cycles~\cite{Rote, narrator}. 

The second is virtual counters, which can be implemented by a single-write multiple-read register~\cite{narrator, narrator-pro, wang2024formally} or an append-only ledger~\cite{kaptchuk2020giving, nimble}. 
The former includes ROTE~\cite{Rote} and Narrator~\cite{narrator}, which adopt a two-phase broadcast protocol.
The latter can be realized by blockchain~\cite{kaptchuk2020giving} or CFT consensus as adopted in Nimble~\cite{nimble}. However, using them in confidential BFT consensus protocols introduces several communication steps. 
In this paper, we use Narrator-Pro, \ie, a single-write multiple-read register~\cite{narrator-pro}, as the software-based counters. 
Despite their high costs, we employ counters only for infrequently updated metadata, thereby avoiding protocol overhead for log commitment.

\bheading{Inc-store consistency dilemma.} There are two fundamental operations when using a trusted counter: incrementing the counter and sealing the state. Because these operations cannot be executed atomically, their sequence leads to two distinct patterns. 

\begin{packeditemize}
    \item \textit{Inc-then-store pattern~\cite{Rote, narrator, narrator-pro, engraft}.}
    In this approach, the counter is incremented before the state is sealed. This guarantees that no previously sealed state can be replayed, since the counter always moves forward. However, if a crash occurs after the counter has been incremented but before the state is sealed, the counter and state become permanently inconsistent, making recovery impossible. From the recovering node’s perspective, a mismatch between the counter value and the binding counter in the sealed state is indistinguishable from either a benign crash or a deliberate rollback attack.

    \item \textit{Store-then-inc pattern~\cite{Memoir, Ariadne}.} The state is sealed first, and then the counter is incremented. This avoids unrecoverable crashes, since the sealed state always exists even if the counter increment fails. However, this introduces a rollback window: an adversary can seal multiple different states under the same counter value and later replay an old one. After a reboot, the node cannot determine which state with the same counter value is most recent, allowing its state to be rolled back.
\end{packeditemize}

\subsubsection{Restricted Faults} \label{appen:restriction}
\sysname assumes at most $f$ nodes may fail concurrently. Without this assumption, the system may gradually lose liveness as no recovering leader can recover its log from collecting $f+1$ replies. Yet this limitation is not unique to our work. Diskless CFT protocols without stable storage, such as VR~\cite{liskov12vr} and variants of Paxos~\cite{chandra2007paxos, konczak2011jpaxos}, also share this constraint (no more than f crashed nodes concurrently). Moreover, all BFT protocols have a security threshold $f$. An adversary compromising more than $f$ nodes would disrupt system correctness. This also holds true for \sysname.

\subsubsection{Customized CFT Consensus Protocols} \label{appen:customization}
Except for rollback issues, 
Wang~\etal~\cite{engraft} also identify several safety and liveness violations of directly porting the Raft protocol into TEEs. 
To address these violations, Wang~\etal propose several countermeasures, including file encryption, network encryption and authentication, and malicious leader detection. 
In this paper, we focus on rollback-resilient recovery and so assume a customized CFT protocol with the above countermeasures running within TEEs. 
In other words, the customized CFT protocol can guarantee safety and liveness properties without rollback attacks.

\subsection{Additional Evaluation} \label{appen:addExp}

\bheading{Throughput vs. Latency.} 
Fig.~\ref{fig:scalability} illustrates the latency of the \sysname-B and its counterparts with increasing throughput until system saturation in both LAN and WAN deployments. 
With 10 faulty nodes, \sysname reaches maximum throughputs of 135.8 kTPS in LAN and 22.6 kTPS in WAN, closely matching or surpassing the best performance among all counterparts. Notably, in LAN, \sysname-B even outperforms Braft-DR, as asynchronous disk writes reduce persistence bottlenecks. Overall, the results confirm that \sysname introduces negligible overhead to log consensus and can even enhance performance in certain settings.

\begin{figure}[h]
  \centering
  \subfloat[f = 10, LAN]{%
      \includegraphics[width=0.48\columnwidth]{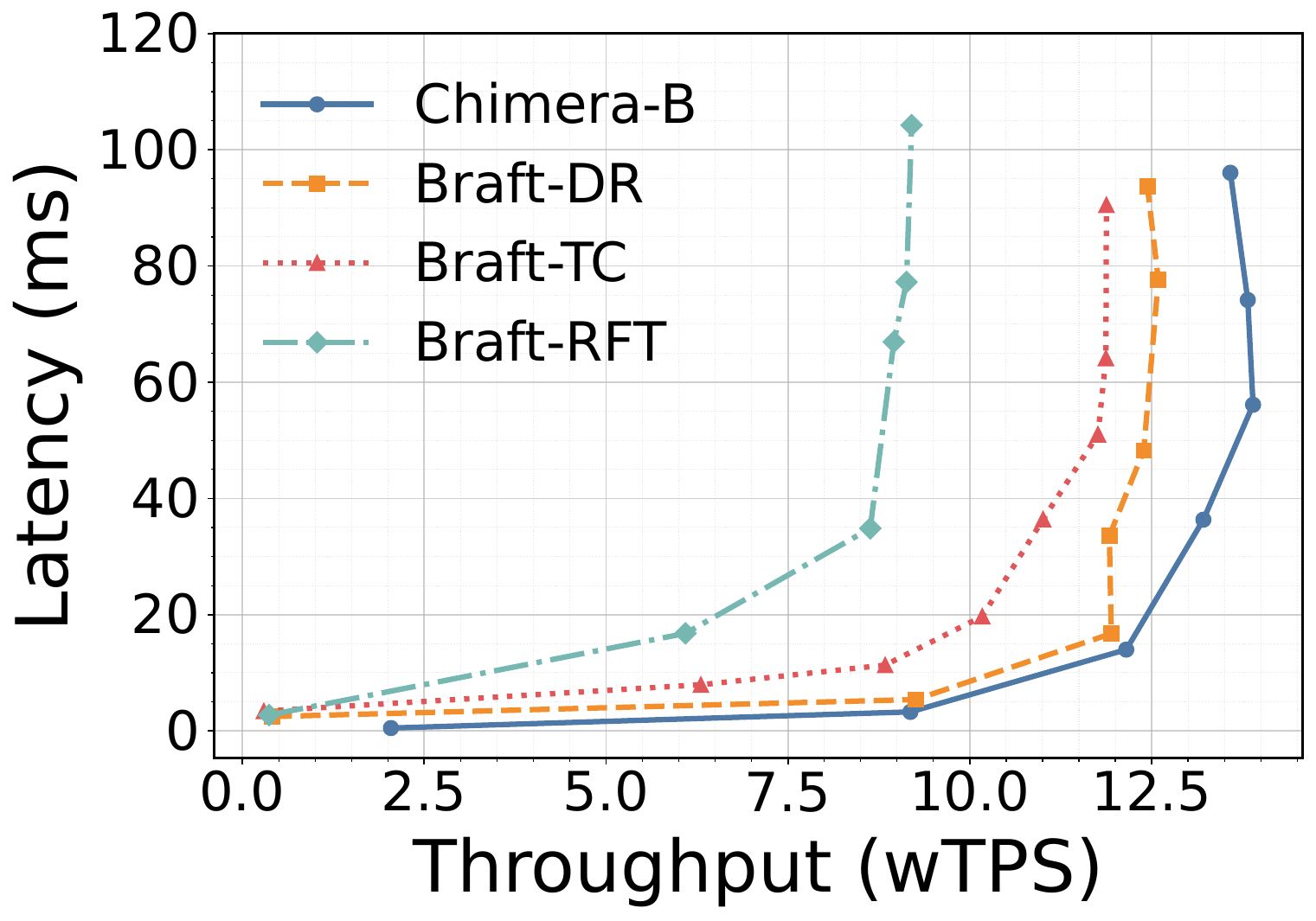}%
      \label{fig:suba}%
  }\hfill
  \subfloat[f = 10, WAN]{%
      \includegraphics[width=0.48\columnwidth]{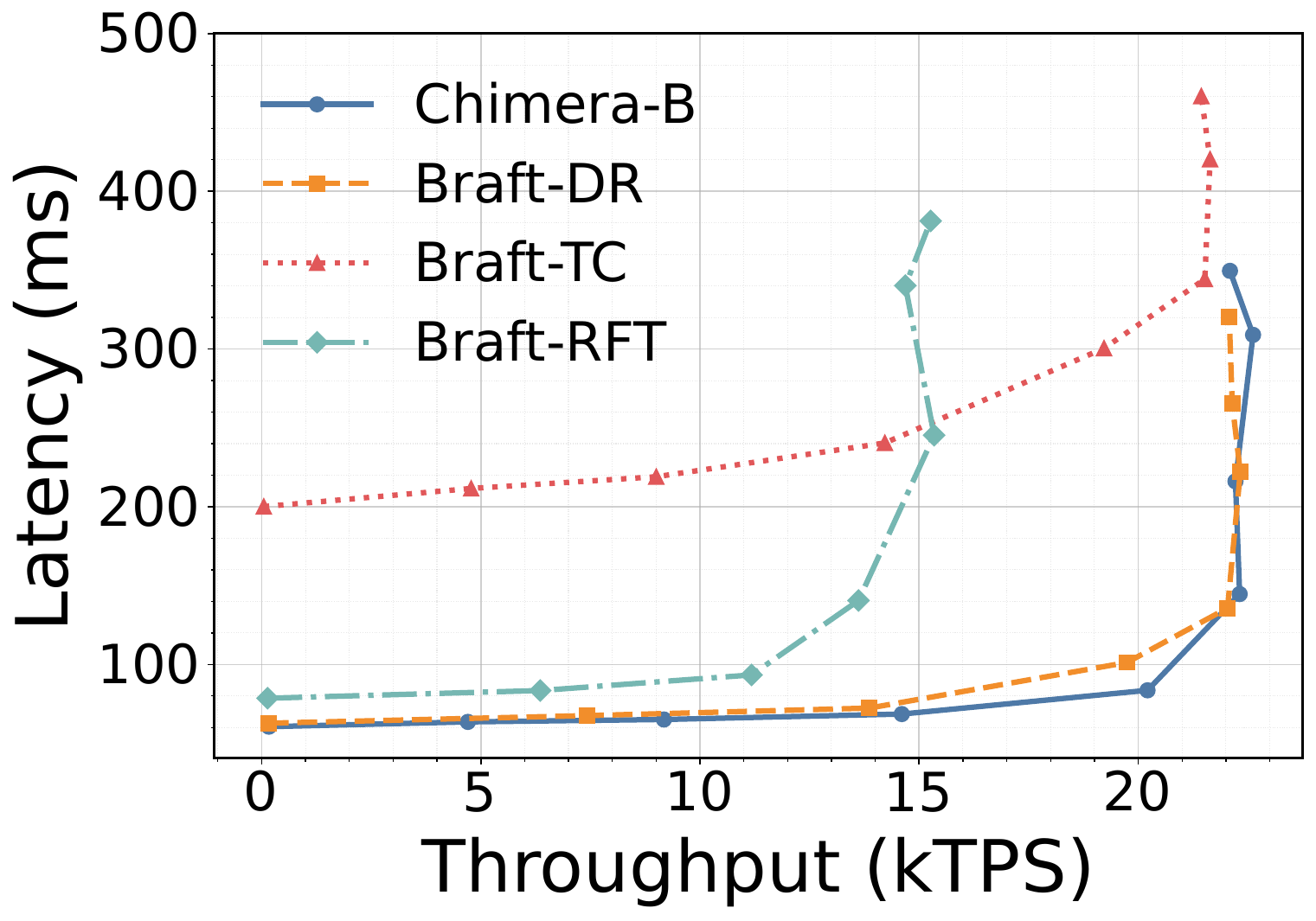}%
      \label{fig:subb}%
  }
  \caption{Throughput vs. Latency of \sysname-B and its counterparts.}
  \label{fig:scalability}
\end{figure}

\bheading{Throughput of Braft-RC under Recovering Faults.}
Fig.~\ref{fig:faultRC} shows the throughput variation during the recovery of faulty nodes in Braft-RC. The throughput behavior at the shutdown point is the same as that observed in \sysname-B. These effects follow the same reasoning as discussed earlier and are not elaborated here.

A key difference in Braft-RC is that recovery is realized through reconfiguration. Newly added nodes do not reload logs from disk but instead start synchronizing directly from the leader. During this synchronization phase, replication traffic competes with ongoing client requests, which slows down replication processing and temporarily reduces throughput until the synchronization completes.

\begin{figure}[h]
  \centering
  \includegraphics[width=\linewidth]{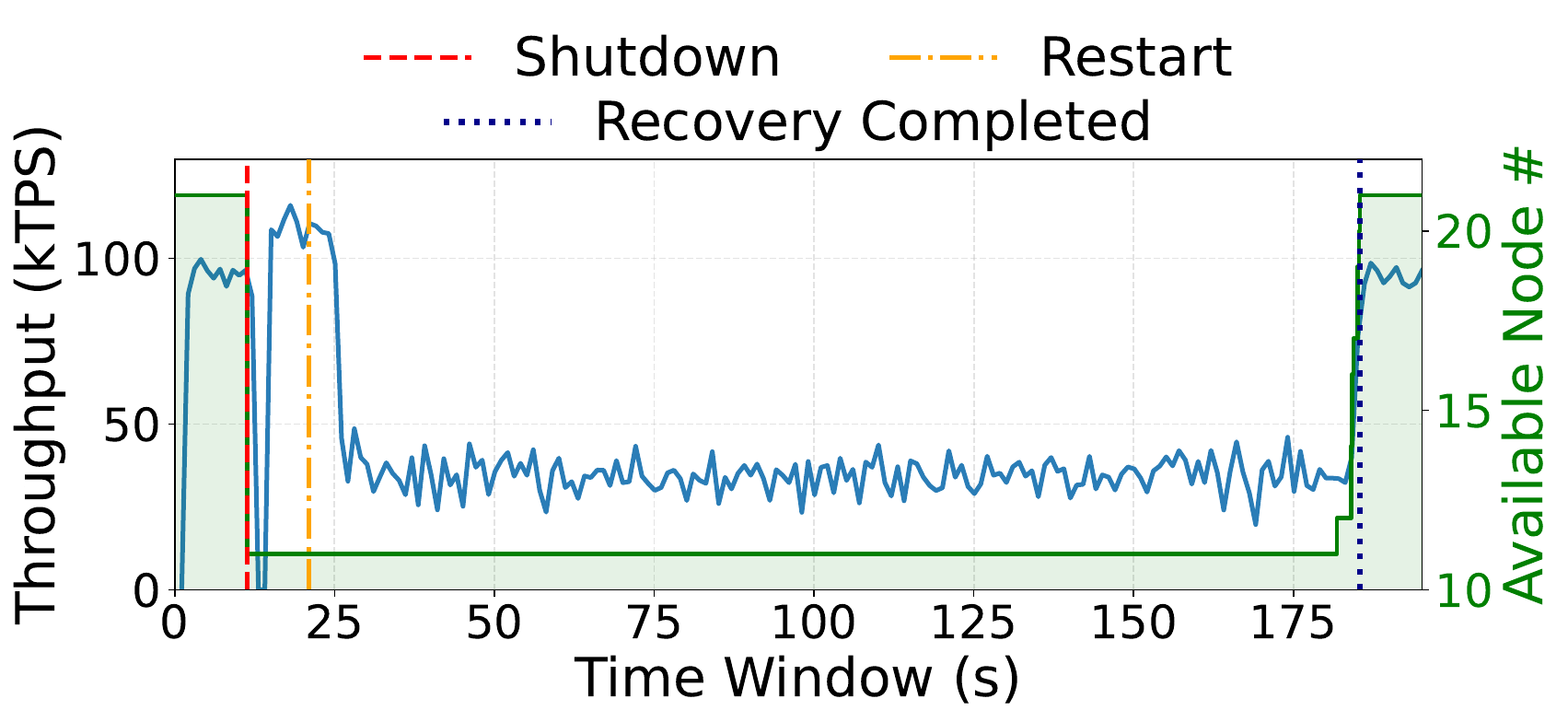} 
  \caption{The fault recovery process of Braft-RC.}
  \label{fig:faultRC}
\end{figure}

\balance
\subsection{Formal Modeling and Analysis} \label{appen:maude}
Our formal specification of \sysname-B consists of approximately 910 LoC in Maude.\footnote{The Maude specification is available at \url{https://github.com/Artifacts2026/CHIMERA/MaudeSpec}.} 
Our modeling follows Agha's \emph{actors} paradigm~\cite{DBLP:books/daglib/0066897}. Specifically, nodes are modeled as actors, and their communication is captured through message passing.
Upon receiving a message, a node may update its local state and possibly generate new messages. The overall system evolves through such message-triggered transitions.

\begin{table}
    \caption{Model checking results.}
    \label{tab:maude}
    \centering
    \scalebox{0.95}{
    \begin{tabular}{@{}lcccc@{}}
        \toprule[1pt]
        \textbf{Property} & \textbf{Metric} & \textbf{3 Nodes} & \textbf{3 Nodes w/ 2 reboots} & \textbf{5 Nodes} \\
        \midrule
        \multirow{2}{*}{P1} 
            & \#States & / & 2,313,191 [45] & 2,200,414[26] \\
            & Time   & 557s & 1,174s & 4,209s \\
        \midrule
        \multirow{2}{*}{P2} 
            & \#States & / & 2,691,520 [45] & 2,652,353 [26] \\
            & Time   & 1,294s & 2,883s & 9,814s \\
        \midrule
        \multirow{2}{*}{P3} 
            & \#States & / & 2,691,520 [45] & 2,652,353 [26] \\
            & Time   & 1,145s & 3,274s & 10,213s \\
        \midrule
        \multirow{2}{*}{P4} 
            & \#States & / & 2,691,520 [45] & 2,652,353 [26] \\
            & Time   & 1,146s & 3,369s & 10,469s \\
        \bottomrule[1pt]
    \end{tabular}}
\end{table}

We verify \sysname-B under TEE rollbacks using linear-temporal-logic (LTL) model checking, focusing on four key properties:

\begin{packeditemize}
\item \textbf{P1.} \emph{Election Safety:} at most one leader can be elected in any epoch.
\item \textbf{P2.} \emph{Leader Completeness:} once a log entry is committed by a leader, any subsequent leader contains this entry.
\item \textbf{P3.} \emph{Recovery Liveness:} a node undergoing reboot eventually completes its recovery procedure.
\item \textbf{P4.} \emph{Protocol Liveness:} every client request is eventually committed as a log entry.
\end{packeditemize}

Table~\ref{tab:maude} shows the model checking results for three cases: three nodes with one reboot, three nodes with two reboots, and five nodes. 
Note that the state space grows rapidly with additional nodes, making exhaustive verification infeasible within a reasonable time, a well-recognized challenge in the formal verification of distributed protocols \cite{DBLP:conf/tacas/LiuOZWM19,10.1145/3689031.3696069,DBLP:conf/icfem/LiuRSGM14,DBLP:journals/pvldb/SchvimerDH20,10.14778/3718057.3718065}. To address this, we utilize Maude's bounded search, which explores the state space
 with depth limits (shown in [square brackets] in Table~\ref{tab:maude}). This enables model checking up to a specified bound. 
Within the explored bounds, Maude reports no counterexamples to the four properties.

\end{document}